\documentclass[sigconf,nonacm]{acmart}

\renewcommand\footnotetextcopyrightpermission[1]{}

\usepackage{graphicx}
\usepackage{float} 
\usepackage{subcaption}
\usepackage{caption}
\usepackage[titlenumbered,ruled,linesnumbered,noend]{algorithm2e}
\usepackage{bm, bbm, mathtools, booktabs}
\usepackage{amsthm}
\usepackage[normalem]{ulem}
\usepackage{enumitem}
\setlist[itemize]{leftmargin=2em}

\usepackage{framed}
\usepackage{tcolorbox}
\usepackage{tabularx}
\usepackage{hyperref}

\newcommand{\kw}[1]{{\ensuremath {\mathsf{#1}}}\xspace}

\newcommand{\etitle}[1]{\vspace{0.3ex}\noindent{\em\underline{#1}}}
\newcommand{\stitle}[1]{\vspace{1.2ex}\noindent{\bf #1}}

\setlist{leftmargin=5.5mm}

\tcbset{
    myboxstyle/.style={
        colback=gray!9,
        colframe=black!70,
        boxrule=1.3pt,
        width=\linewidth,
        arc=2mm,
        auto outer arc,
        boxsep=.5mm,
        left=1mm,
        right=1mm,
        top=.5mm,
        bottom=.5mm
    }
}

\newcommand{\eop}{\hspace*{\fill}\mbox{$\Box$}}     
\newcounter{example}
\renewcommand{\theexample}{\arabic{example}}
\newenvironment{example}{
        \vspace{1.2ex}
        \refstepcounter{example}
        {\noindent\bf Example \theexample:}}{
        \eop\vspace{1.2ex}}

\newcommand{\nthesection}{\arabic{section}}
\newcounter{theorem}
\renewcommand{\thetheorem}{\arabic{theorem}}
\newcounter{prop}
\renewcommand{\theprop}{\arabic{theorem}}
\newcounter{lemma}
\renewcommand{\thelemma}{\arabic{theorem}}
\newcounter{cor}
\renewcommand{\thecor}{\arabic{theorem}}
\newenvironment{theorem}{\begin{em}
        \refstepcounter{theorem}
        {\vspace{1ex} \noindent\bf  Theorem  \thetheorem:}}{
        \end{em}\eop\vspace{1ex}} 
\newcounter{definition}
\renewcommand{\thedefinition}{\arabic{definition}}
\newenvironment{definition}{
        \vspace{1.5ex}
        \refstepcounter{definition}
        {\noindent\bf Definition {\bf \thedefinition}:}}{\eop\vspace{1.5ex}
}
\newcounter{alg}[section]
\renewcommand{\thealg}{\nthesection.\arabic{alg}}

\newcounter{arule}
\renewcommand{\thearule}{\arabic{arule}}

\newcounter{claim}
\renewcommand{\theclaim}{\arabic{claim}}

\renewenvironment{proof}{
        \vspace{1ex}
        {\noindent\bf Proof:}}{\eop\vspace{1ex}}

\newcommand{\eg}{\emph{e.g.},\xspace} 
\newcommand{\ie}{\emph{i.e.},\xspace} 

\newenvironment{proofS}{
        \vspace{1ex}
        {\noindent\bf Proof sketch:\ }}{\eop\vspace{1ex}}

\newlist{myitemize}{itemize}{3}
\setlist[myitemize,1]{label=$\circ$,leftmargin=2.8ex}
\setlist[myitemize,2]{label=$\bullet$,leftmargin=2.8ex}
\setlist[myitemize,3]{label=$\diamond$, leftmargin=2.2ex}
\newenvironment{mbi}{\begin{myitemize}
        \setlength{\topsep}{0.1ex}\setlength{\itemsep}{0.1ex}\vspace{0.1ex}}
        {\end{myitemize}\vspace{0.1ex}}

\newcommand{\mei}{\end{myitemize}\vspace{0.2ex}}

\newcommand{\eat}[1]{}

\newcommand{\C}{{\mathcal C}}

\newcommand{\RAC}{\kw{RAC}}
\newcommand{\LLM}{\kw{LLM}}

\newcommand{\LRU}{\kw{LRU}}

\newif\ifreport
\reportfalse
\begin{document}

\title{\RAC: Relation-Aware Cache Replacement for Large Language Models}

\author{Yuchong Wu}
\affiliation{%
  \institution{Zhejiang University}
  \city{Hangzhou}
  \country{China}
}
\email{wychong@zju.edu.cn}

\author{Zihuan Xu}
\affiliation{%
  \institution{Shenzhen Institute of Computing Sciences}
  \city{Shenzhen}
  \country{China}
}
\email{xuzihuan@sics.ac.cn}

\author{Wangze Ni}
\affiliation{%
  \institution{Zhejiang University}
  \city{Hangzhou}
  \country{China}
}
\email{niwangze@zju.edu.cn}

\author{Peng Cheng}
\affiliation{%
  \institution{Tongji University}
  \city{Shanghai}
  \country{China}
}
\email{cspcheng@tongji.edu.cn}

\author{Lei Chen}
\affiliation{%
  \institution{Hong Kong University of Science and Technology}
  \city{Hong Kong}
  \country{China}
}
\email{leichen@cse.ust.hk}

\author{Xuemin Lin}
\affiliation{%
  \institution{Shanghai Jiao Tong University}
  \city{Shanghai}
  \country{China}
}
\email{xuemin.lin@sjtu.edu.cn}

\author{Heng Tao Shen}
\affiliation{%
  \institution{Tongji University}
  \city{Shanghai}
  \country{China}
}
\email{shenhengtao@tongji.edu.cn}

\author{Kui Ren}
\affiliation{%
  \institution{Zhejiang University}
  \city{Hangzhou}
  \country{China}
}
\email{kuiren@zju.edu.cn}

\begin{abstract}
The scaling of Large Language Model (\LLM) services faces significant cost and latency challenges, making effective caching under tight capacity crucial. Existing cache replacement policies, from heuristics to learning-based methods, predominantly rely on limited-window statistics such as recency and frequency. We show these signals are not robust for real-world \LLM workloads, which exhibit long reuse distances and sparse local recurrence.

To address these limitations, we propose \emph{Relation-Aware Cache} (\RAC), an online eviction strategy that leverages semantic relations among requests to guide eviction decisions.
\RAC synthesizes two relation-aware signals: (1) \textbf{Topical Prevalence}, which aggregates access evidence at the topic level to capture long-horizon reuse; and (2) \textbf{Structural Importance}, which leverages local intra-topic dependency structure to discriminate entries by their future reuse value. Extensive evaluations show that \RAC maintains high effectiveness across diverse workloads, consistently surpassing state-of-the-art baselines by 20\%--30\% in cache hit ratio.
\end{abstract}

\maketitle

\section{Introduction}
\label{sec:introduction}
Large language model (\LLM) services have been widely deployed at scale, operating under growing request concurrency and context lengths~\cite{openai2022chatgpt,anthropic2023claude}. 
To reduce the cost of transformer-based inference, recent systems increasingly rely on {\em caching} to reuse previously generated content (\ie~semantic cache)~\cite{bang2023gptcache,li2024scalm,gillmeancache} or intermediate states (\ie~\kw{KV} cache)~\cite{lin2025ragcache,cacheblend2024}. 
However, real-world deployments operate under tight cache budgets, making trace-driven admission and eviction decisions a central challenge for maximizing future reuse~\cite{yang2001kddcache}.

Existing \LLM\ cache replacement methods mainly fall into two categories.
(1) \emph{Traditional policies}~\cite{mattson1970evaluation, lee2001lrfu, megiddo2003arc, Cao1997GDS, Cherkasova1998GDSF, Beckmann2018LHD, liu2023s3fifo, Zhang2023Sieve, zhang2023lfru, hu2024qdfifo} 
compute an eviction score from short-window statistics (\eg recency, frequency, or their combinations) and evict the entry with the minimum score under capacity pressure.
(2) \emph{Learning-based methods}~\cite{song2020lrbelady,liu2020ilcache,Vietri18,Rodriguez21,song2023halp,yang2023lrb} 
learn a predictor from historical access traces to estimate future reuse (\eg reuse within a horizon or next access distance) and evict entries according to the predictions.
\looseness=-1

However, recent studies show that real \LLM workloads exhibit long reuse distances and sparse local recurrence~\cite{yu2025smartcache}.
As a result, most cache entries are accessed only once or reappear after long gaps, leaving little reuse signal within short observation windows.
Under such workloads, traditional policies that rely on short-range statistics become ineffective, as recency and frequency no longer correlate with future reuse.
Learning-based methods, constrained by finite prediction horizons, similarly fail to capture reuse events beyond their observable scope.
We next use a concrete example to illustrate how such failures arise in practice.

\begin{table}[t]
\centering
\footnotesize
\setlength{\tabcolsep}{6pt}
\renewcommand{\arraystretch}{1.1}
\begin{tabular}{c p{6.1cm}}
\toprule
\textbf{Items} & \multicolumn{1}{c}{\textbf{Request semantics}} \\
\midrule
\parbox[t]{1.3cm}{\centering
(\textbf{context})\\
$a_0$\\[2pt]
(queries)\\
$a_1$--$a_5$
}
&
\parbox[t]{6.1cm}{%
A code-centric request sequence.
$a_0$ (\textbf{context}) asks the \LLM to read and explain a given code snippet.
$a_1$ queries the functionality of a specific function;
$a_2$ queries variable roles and data flow;
$a_3$ queries invariants or assumptions;
$a_4$ queries potential corner cases;
and $a_5$ queries possible optimizations.%
}
\\
\hline
\parbox[t]{1.3cm}{\centering
(queries)\\
$a_1^{*}$--$a_5^{*}$
} 
&
\parbox[t]{6.1cm}{
A set of follow-up queries reusing the code context in $a_0$,
where $a_1^*$--$a_5^*$ target different functions, variables, assumptions,
corner cases, and optimizations of the same code. 
}
\\
\hline
\parbox[t]{1.3cm}{\centering
(\textbf{context})\\
$b_2$\\[2pt]
(queries)\\
$b_0$, $b_1$, $b_3$--$b_5$
}
&
\parbox[t]{6.1cm}{
A writing task sequence.
$b_0$ asks the \LLM to draft a short text on a given topic.
$b_1$ asks for revision suggestions.
$b_2$ (\textbf{context}) specifies explicit writing constraints
(\eg tone and formatting rules).
$b_3$ asks to rewrite the text under the constraints;
$b_4$ asks to condense it;
and $b_5$ asks to finalize the text accordingly.
}
\\
\hline
\parbox[t]{1.3cm}{\centering
(queries)\\
$b_0^{*}, b_1^{*}$\\
$b_3^{*}$--$b_5^{*}$
}
&
\parbox[t]{6.1cm}{%
A repeated writing task that reuses the same constraints in $b_2$.
$b_0^{*}$ and $b_1^{*}$ draft a new text on a different topic, while
$b_3^{*}$--$b_5^{*}$ request revision, rewriting, condensation,
and finalization under the same constraints specified in $b_2$.
}
\\
\bottomrule
\end{tabular}
\caption{Example items and their request semantics.}
\label{tab:example_prompts}
\end{table}

\begin{figure}[t]
  \centering
  \includegraphics[width=\columnwidth]{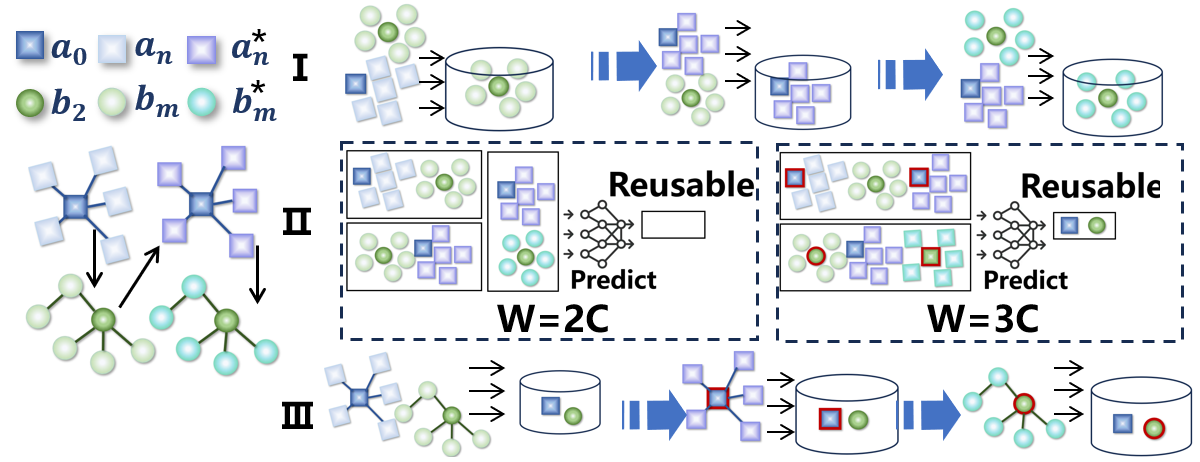}
  \caption{Demonstration of traditional, learning-based, and offline-optimal policies on Example~1.}
  \label{fig:example1}
\end{figure}

\begin{example}
\label{ex:core_branch_batch}
Consider the \LLM request sequence
$\{a_0\!\sim\! a_5\}
\!\rightarrow\! \{b_0\!\sim\! b_5\}
\!\rightarrow\! \{a_0, a_1^{*}\!\sim\! a_5^{*}\}
\!\rightarrow\! \{b_0^*\!\sim\! b_5^{*}\}$,
where each item corresponds to a query listed in Table~\ref{tab:example_prompts}, and the sequence alternates between two topics (\ie A: coding and B: writing), each forming a coherent task in which context-setting requests (\ie $a_0$ and $b_2$) are followed by related queries.
Figure~\ref{fig:example1} illustrates the semantic structure among requests induced by semantic relatedness. Assume a cache of size $|\C|=6$; we compare the outcomes of different cache strategies.

\etitle{Traditional policies}.
As shown in Fig.~1(I), taking \LRU~\cite{mattson1970evaluation} as an example, each batch of semantically related requests fills the cache.
The arrival of the next batch evicts all resident entries before any reuse can occur, resulting in zero cache hits.

\etitle{Online learning policies}.
Taking \kw{LRB}~\cite{song2020lrbelady} as an example, at cold start (Fig.~1(I)) or with a small training window (Fig.~1(II-left)), no reuse is observed and the policy degenerates to traditional behavior.
Only with a sufficiently large window (Fig.~1(II-right)) can reuse be learned, at the cost of significantly higher overhead  (especially for semantic caches, where hit determination itself requires costly similarity computation).

\etitle{Offline optimal}.
As shown in Fig.~1(III), the offline optimal policy preserves structurally central entries while evicting peripheral ones.
Across topic switches, core requests (\eg $a_0$ and $b_2$) are retained and repeatedly reused, while peripheral entries are trimmed, maximizing reuse by preserving structurally central items.
\end{example}

The example above shows that under tight cache capacity, many entries are evicted before they are reused but offline optimal shows it can be reused.
As a result, short-window signals such as recency and access frequency provide weak indications of future reuse.
Nevertheless, we observe that reuse in \LLM can depend on request relations such as topical recurrence across episodes and prerequisite dependencies within a topic.

These observations raise three questions: (1) what relations among requests are observable \emph{online} with low overhead; (2) how to convert such relations into an eviction signal beyond entry-local recency/frequency; and (3) how to implement the resulting policy under a hard capacity constraint.

To address them, we (1) operationalize the two patterns as two online eviction signals: \textbf{Topical Prevalence (TP)} aggregates long-horizon semantic recurrence at the topic level, while \textbf{Topic Structural Importance (TSI)} prioritizes within-topic context anchors induced by prerequisite dependency; (2) integrate TP and TSI into a unified eviction value via a lightweight heuristic; and (3) implement the resulting policy under a hard capacity budget.
Both signals are lightweight to maintain online, yield relation-derived cues beyond recency/frequency, and are directly actionable under hard-capacity replacement.

\stitle{Contributions \& Organization.}
We propose Relation-Aware Cache (\RAC), an online cache eviction strategy that uses relation-aware values and provides an end-to-end pipeline from relation capture to eviction decisions.

\etitle{Relation-aware eviction rule. (Section~\ref{sec:basic-idea})}
We formalize semantic cache replacement under a capacity constraint and show why entry-local short-window signals are insufficient when reuse is delayed.
Based on this model, we introduce two relation-aware values and combine them into a single eviction score, yielding an online rule that specifies both cache updates and eviction under pressure.

\etitle{Computation of Topical Prevalence. (Section~\ref{sec:component-A})} 
We design an online mechanism to estimate topical prevalence by organizing requests into topic clusters and aggregating access evidence at the topic level.
The resulting topic score summarizes long-horizon reuse signals shared by entries in the same topic, providing information that is missing from per-entry recency/frequency.

\etitle{Computation of Structural Importance. (Section~\ref{sec:component-B})} 
We design an online module to estimate structural importance using a lightweight identifier of local dependency structure.
It assigns higher value to prerequisite entries that support more downstream requests, complementing purely statistical signals when direct repeats are sparse.

\etitle{Experimental study. (Section~\ref{sec:evaluation})}
We evaluate \RAC on real-world dialogue traces and on synthetic workloads derived from the modeled process.
On real traces, \RAC improves \textbf{hit ratio} by \textbf{20\% on average} over representative frequency-/recency-based baselines; on model-driven synthetic workloads, it achieves an \textbf{average 30\% gain} in regimes where relation-aware aggregation is most beneficial.

We introduce the problem setting in Section~\ref{sec:preliminary}, discuss related work in Sections~\ref{sec:related-work}, and conclude in Section~\ref{sec:conclusion}.

\section{Preliminary and Problem Formulation}
\label{sec:preliminary}

We first define the workload and cache abstractions and then formalize the online cache admission and eviction problem.

\stitle{Preliminary}. We specify (1) a topic-aware workload model for dialogue queries and (2) a cache model that maps queries to cache entries and defines cache hits.

\etitle{Topic.}
Following \cite{arguello2006topicseg,groszsidner1986discourse}, we define a \emph{topic} as a dialogue segment with a relatively stable semantic focus (\ie a consistent \emph{discourse segment purpose}, DSP). We model the dialogue as a query sequence indexed by time steps $t=1,2,\ldots$. Each query $q_t$ is assigned a topic label $Z_t\in\{1,\ldots,S\}$, where $S$ is the number of distinct topics; when $Z_t=s$, query $q_t$ belongs to topic $s$. A topic $s$ may appear in multiple \emph{topic episodes}, where each episode is a maximal contiguous time interval during which $Z_t=s$. Episodes of the same topic can be separated by other topics, allowing topic $s$ to reappear later.

\etitle{Topic sequence as a semi-Markov process.}
Given the topic labels $\{Z_t\}$, we obtain an episode-level sequence by collapsing each topic episode into a single state visit. Each visit is represented by (1) the topic identity and (2) its episode length (the number of consecutive queries in that episode). For example, if the topic-label sequence is $\{A,A,B,B,B,A\}$, then the episode-level sequence is $(A,2)\!\to\!(B,3)\!\to\!(A,1)$.
We model this episode-level sequence as a semi-Markov process over topics, where:
(1) transitions occur only at episode boundaries;
(2) the sojourn time in a state equals the episode length; and
(3) the next topic depends only on the current topic (\ie a first-order Markov property for the embedded jump chain).\looseness=-1

\etitle{Intra-topic query dependency.}
Following \cite{li2020molweni,li2014discourse}, we assume that queries within a topic episode exhibit an implicit dependency structure. For each topic $s\in\{1,\ldots,S\}$ and any episode of topic $s$, we represent the intra-episode dependencies as a set of time-respecting \emph{discourse dependency links} $\mathcal{E}_s$. Each link $(i,j)\in\mathcal{E}_s$ with $i<j$ indicates that query $q_j$ depends on an earlier query $q_i$ for context (\eg co-reference resolution, clarification, or continuation). Therefore, queries in an episode induce a directed acyclic graph (DAG) that captures the flow of context. Unless otherwise specified, we define and use these dependencies only within the same topic episode.

\etitle{Cache.}
We next define the cache model used in this paper, including (1) the cache entry abstraction, (2) the cache store $\C$ and its capacity, and (3) the cache hit criterion.
\mbi
\item \textbf{Entry.} A cache entry $e$ is the atomic object managed by the cache. Each entry has a semantic embedding, and we use a similarity function $\kw{sim}(\cdot,\cdot)$ (\eg cosine similarity) to measure proximity between embeddings. A user query $q_t$ may map to one entry (query-level caching) or multiple entries (chunk-level caching).

\item \textbf{Store.} We denote the cache as a set of entries $\C=\{e_1,\dots,e_m\}$ with capacity $|\C|\le C$. Depending on the cache type, an entry may store reusable content in one of the following forms: (a) semantic content (\eg passages, past responses, summaries, or prompt patches~\cite{bang2023gptcache}), (b) \kw{KV} states for prefill reuse~\cite{yang2024memory3,yang2025kvshare}, or (c) a hybrid payload that jointly manages text and \kw{KV} states~\cite{li2025memos}. Each entry also maintains lightweight intrinsic metadata (\eg recency and frequency).

\item \textbf{Hits.} A query $q_i$ is a \emph{hit} on entry $e_j$ if it satisfies the system-defined equivalence condition, typically via one of the following: (a) \emph{semantic equivalence}, where $\kw{sim}(q_i,e_j)\ge \tau$ for a predefined threshold $\tau$; or (b) \emph{content equivalence}, where the cached payload in $e_j$ matches the required content or aligns with the query context (\eg prefix alignment in \kw{KV} caches). Otherwise, $q_i$ is a miss.\looseness=-1
\mei

\vspace{-1.5ex}
\stitle{Problem.} We are now ready to formalize the \emph{online cache admission and eviction problem} for maximizing cache hits.
\mbi
\item \emph{Input}: A time-ordered \LLM\ query stream $Q\!=\!\{q_1,q_2,\ldots\}$; a cache of capacity $C$ that is initially empty; and the cache hit criterion.\looseness=-1
\item \emph{Output}: An online caching policy $\pi$ that decides, upon each request arrival, (1) whether to admit the corresponding entry into the cache and (2) which entries to evict when the cache is full.
\item \emph{Objective}: Maximize the total number of cache hits over $Q$ subject to the capacity constraint $C$.
\mei
\vspace{-0.5ex}

\etitle{Remark.} Observe the following.
(1) This problem is fully online: $\pi$ makes decisions without future knowledge and can only use information available up to the current time step.
(2) The hit criterion is system-defined, and our formulation is agnostic to the specific cache type.
In other words, our method can be instantiated across mainstream caching settings today, such as embedding-based semantic equivalence for semantic caching ~\cite{bang2023gptcache} and compositional content equivalence for \kw{KV} caching in LLM serving. \cite{Agarwal2025CacheCraft, lin2025ragcache, cacheblend2024}
(3) The topic-aware workload abstractions (topics, episodes, and intra-episode dependencies) specify additional structure in $Q$ that our policy can exploit, but they do not change the underlying objective of maximizing cache hits.

\section{Design of Relation-Aware Caching}
\label{sec:topic-value}

This section designs the relation-aware cache value used by the online admission and eviction policy. Section~\ref{sec:basic-idea} introduces the basic idea and defines the entry value as the product of two online signals, where (1) topical prevalence captures topic-level activeness over time and (2) structural importance captures how critical an entry is as a context anchor within a topic episode. Section~\ref{sec:component-A} and Section~\ref{sec:component-B} detail the online computation and maintenance of topical prevalence and structural importance, respectively.

\subsection{Basic Idea}
\label{sec:basic-idea}

As in Section~\ref{sec:preliminary}, we model the workload as a query sequence $\{q_t\}$ with topic labels $\{Z_t\}$ ~\cite{jelenkovic2009asymptotic}.
We further interpret the topic-label process at episode granularity as a semi-Markov process.

\etitle{Topic occupancy and query generation.}
For each topic $s\in\{1,\ldots,S\}$, let (1) $\pi_s$ denote its long-run occupancy probability, \ie the steady-state fraction of time steps with topic label $Z_t=s$, and (2) $p(q\mid s)$ denote the conditional probability of observing the concrete query content $q$ when the topic label is $s$.
By the law of total probability, the marginal probability of observing a specific query $q_t$ at time $t$ is\looseness=-1
\[
p(q_t)\;\triangleq\;\sum_{s=1}^{S}\pi_s\,p(q_t\mid s).
\]

Equivalently, isolating the in-topic component for the realized topic label $Z_t$ yields
\[
p(q_t)
=
\pi_{Z_t}p\big(q_t\mid Z_t\big)
+
\sum_{s\neq Z_t}\pi_s\,p(q_t\mid s).
\]
Due to topic locality, we typically have $p(q_t\mid Z_t)\gg p(q_t\mid s)$ for $s\neq Z_t$, so the in-topic term $\pi_{Z_t}p\big(q_t\mid Z_t\big)$ dominates the mixture.

\etitle{Relation-aware value.}
\noindent
We design a heuristic value to estimate the factorized term $\pi_{Z_t} p(q\mid Z_t)$ by tracking its two factors with observable online signals:
\mbi
  \item \textbf{Topical Prevalence (\kw{TP}).} A per-topic temporal signal $\kw{TP}(s)$ tracks the topic weight $\pi_s$ via a lightweight decay-and-accumulate update on topic hits (see details in Section~\ref{sec:component-A}).
  \item \textbf{Topic Structural Importance (\kw{TSI}).} A per-item topological signal $\kw{TSI}(q)$ proxies the in-topic strength $p(q\mid s)$ by tracking the dependency in-degree (\ie dependency count) of item $q$ within its topic episode (see details in Section~\ref{sec:component-B}).
\mei

Consequently, to approximate the dominant in-topic component $\pi_Z p(q\mid Z)$ using these online statistics, we define the unified heuristic value for a query $q$ associated with topic $Z$ as

\begin{equation}
\kw{Value}(q) = \kw{TP}(Z)\cdot \kw{TSI}(q).
\end{equation}

\etitle{Online flow.}
With the relation-aware value $\kw{Value}(\cdot)$ in place, we next design a complete cache management workflow to make it actionable online. The workflow integrates two computation components that maintain the required online signals, and specifies how these signals are refreshed and used to guide eviction (Algorithm~\ref{alg:tp_tsi_mainflow}). Upon the arrival of request $q_t$, we:
(1) refresh $\kw{TP}(Z_t)$;
(2) update $\kw{TSI}(q)$ for items $q$ that are structurally related to $q_t$ within the current topic episode; and
(3) insert the corresponding cache entry into $\C$ and, if $|\C|>C$, evict the entry whose associated item $q$ has the smallest $\kw{Value}(q)$.

\begin{algorithm}[t]
\caption{Cache-Side Main Workflow}
\label{alg:tp_tsi_mainflow}
\KwIn{incoming request $q_t$; cache $\C$; current time step $t$}
\KwOut{updated cache $\C$}
\SetKwProg{Fn}{Procedure}{}{end}

\Fn{\textsc{OnArrive}($q_t,\C,t$)}{
  $Z_t,\{\kw{TP}(s)\}\leftarrow \textsc{UpdateTP}(q_t,\C,t)$\tcp*[r]{Alg.~\ref{alg:tp_mainflow}}
  $\{\kw{TSI}(q)\}\leftarrow \textsc{UpdateTSI}(q_t,\C,t)$\tcp*[r]{Alg.~\ref{alg:tis-update-lite}}

  insert the cache entries $\bar{e}_t$ of $q_t$ (with label $Z_t$) into $\C$\;
  \While{$|\C|>C$}{
    evict entries $\bar{e}_i\subset\C$ (for query $q_i$ with label $Z_i$) with the minimum $\kw{TP}(Z_i)\cdot \kw{TSI}(q_i)$\;
  }
  \KwRet{$\C$}\;
}
\end{algorithm}

\subsection{Computation of Topical Prevalence.}
\label{sec:component-A}

As a key component of RAC and the first step in its online workflow, we maintain a topic-level \emph{topical prevalence} signal as an online surrogate for the occupancy $\pi_s$, and use it to guide cache admission and eviction. For each topic $s$, we maintain a score $\kw{TP}(s)$ that summarizes how \emph{active} $s$ is at the current time, combining both hit frequency and recency via exponential decay. This score is evaluated during eviction to compare entries across topics.

\begin{definition}[Topical prevalence ($\kw{TP}$)]
\label{def:tp}
For a topic $s$, let $\mathcal{H}_t(s)\triangleq\{i\le t: Z_i=s\}$ denote its hit times up to time $t$.
Given a decay coefficient $\alpha\ge 0$, we define
\[
\kw{TP}_t(s)\;\triangleq\;\sum_{i\in\mathcal{H}_t(s)}\Big(\tfrac{1}{2}\Big)^{\alpha\,(t-i)}.
\vspace{-2ex}\]
\end{definition}

Intuitively, hits are exponentially down-weighted by age, such that $\kw{TP}_t(s)$ captures both recency and frequency.

\etitle{Online maintenance for $\kw{TP}$.}
Directly maintaining $\kw{TP}$ in Definition~\ref{def:tp} requires iterating over all hit times $i\in\mathcal{H}_t(s)$.
To enable $O(1)$ updates, we cache the value at the latest hit.
Specifically, let $t_{\kw{last}}(s)\triangleq \max\mathcal{H}_t(s)$ denote the most recent hit time of topic $s$ (with $t_{\kw{last}}(s)=0$ if $\mathcal{H}_t(s)=\emptyset$), and let $\kw{TP}_{\kw{last}}(s)$ denote the stored $\kw{TP}$ value right after processing the request at time $t_{\kw{last}}(s)$.
For any time $t\ge t_{\kw{last}}(s)$, we have
\[
\kw{TP}_t(s)
=
\Big(\tfrac{1}{2}\Big)^{\alpha\,(t-t_{\kw{last}}(s))}\cdot \kw{TP}_{\kw{last}}(s).
\]

Therefore, we store only two per-topic scalars: $t_{\kw{last}}(s)$ and $\kw{TP}_{\kw{last}}(s)$.
On a hit to $s$, we update $\kw{TP}_{\kw{last}}(s)\leftarrow \kw{TP}_t(s)$ and set $t_{\kw{last}}(s)\leftarrow t$.
During eviction, we compute $\kw{TP}_t(s)$ on demand using the closed form above.
We next describe how the cache performs topic routing and triggers the above $\kw{TP}$ refresh online.

\etitle{Cache-side topic routing and $\kw{TP}$ refresh.}
Algorithm~\ref{alg:tp_mainflow} specifies how the cache assigns each request $q_t$ to a topic and refreshes the corresponding $\kw{TP}$ state.
We maintain a cache-side topic index: each topic $s$ stores (1) a representative embedding $r(s)$ for routing and (2) a list of resident entries currently assigned to $s$.
On arrival of $q_t$, we:
(1) compute its similarity to each $r(s)$ and keep topics within the similarity threshold $\tau$;
(2) if no candidate exists, create a new topic label $Z_t$ and initialize its states;
(3) otherwise set $Z_t$ to the most similar candidate;
(4) treat $Z_t$ as a hit and refresh $\kw{TP}_{\kw{last}}(Z_t)$ by decay-and-increment, and set $t_{\kw{last}}(Z_t)\leftarrow t$.
We detail how to update representatives and maintain membership under eviction in Appendix~\ref{app:topic_mgmt}.\looseness=-1

\begin{algorithm}[t]
\caption{Cache-side Topic Routing and $\kw{TP}$ Refresh}
\label{alg:tp_mainflow}
\KwIn{incoming request $q_t$; cache $\C$; current time step $t$; similarity threshold $\tau$; decay coefficient $\alpha$}
\KwData{states $\{t_{\kw{last}}(s)\}$, $\{\kw{TP}_{\kw{last}}(s)\}$, and representative embedding $\{r(s)\}$ for each appeared topic $s$}
\KwOut{assigned topic label $Z_t$ of $q_t$, and $\{\kw{TP}(s)\}$ for each topic $s$ that appears in $\C$}
\SetKwProg{Fn}{Procedure}{}{end}
\Fn{\textsc{UpdateTP}($q_t,\C,t$)}{
  $Z_t \leftarrow \kw{SearchTopic}(q_t, \tau, \{r(s)\})$\;
  \If{$Z_t=\emptyset$}{
    create a new topic label $Z_t$ and assign $Z_t$ to $q_t$\;
    initialize $t_{\kw{last}}(Z_t)\leftarrow t$, $\kw{TP}_{\kw{last}}(Z_t)\leftarrow 0$, and $r(Z_t)\leftarrow$ the embedding of $q_t$\;
  }

  $\kw{TP}_{\kw{last}}(Z_t)
  \leftarrow
  \big(\tfrac{1}{2}\big)^{\alpha\cdot\big(t-t_{\kw{last}}(Z_t)\big)}
  \cdot \kw{TP}_{\kw{last}}(Z_t) + 1$\;
  $t_{\kw{last}}(Z_t) \leftarrow t$\;

  evaluate $\kw{TP}(s)=\big(\tfrac{1}{2}\big)^{\alpha\cdot\big(t-t_{\kw{last}}(s)\big)}\cdot \kw{TP}_{\kw{last}}(s)$ lazily for topics $s$ that appear in $\C$ during eviction process\;
  \KwRet{$Z_t$ and $\{\kw{TP}(s)\}$}\;
}
\end{algorithm}

\subsection{Computation of Structural Importance.}
\label{sec:component-B}

We next introduce an item-level signal to prioritize entries within the same topic.
Within each episode of topic $s$, queries are connected by dependency links $\mathcal{E}_s$ as defined in Section~\ref{sec:preliminary}.
Intuitively, a query $q$ is important for two reasons: (1) it is requested multiple times within topic $s$, and (2) it provides context required by downstream queries through $\mathcal{E}_s$.
We capture these two effects by using a topic structural importance score $\kw{TSI}(q)$ to guide within-topic eviction, and we formally define $\kw{TSI}(q)$ as follows.

\begin{definition}[Topic Structural Importance (\kw{TSI})]
\label{def:tsi}
For a query $q$ routed to topic $s$, we define
\[
\kw{TSI}(q)\;\triangleq\;\kw{freq}(q)\;+\;\lambda\,\kw{dep}(q),
\]
where $\kw{freq}(q)$ is the number of cache hits to $q$ so far in topic $s$, $\lambda\ge 0$ is a weight parameter, and $\kw{dep}(q)$ is the downstream hit mass induced by dependency links $\mathcal{E}_s$:
\[
\kw{dep}(q_k)\;\triangleq\;\sum_{(q_k,q_j)\in\mathcal{E}_s}\kw{freq}(q_j).
\]
\end{definition}

Intuitively, $\kw{dep}(q_k)$ aggregates the request mass of downstream queries that rely on $q_k$ through $\mathcal{E}_s$.
Therefore, evicting $q_k$ can trigger additional misses on every future access to those dependent queries.
The following theorem formalizes this monotonic relation between eviction cost and $\kw{dep}(q_k)$.

\begin{theorem}
\label{thm:dep_miss}
Within a topic, evicting a query $q_k$ increases the long-run miss probability by an amount that is monotonically increasing in $\kw{dep}(q_k)$.
\end{theorem}

\proofS
Under the dependency links $\mathcal{E}_s$, any downstream query $q_j$ with $(q_k,q_j)\in\mathcal{E}_s$ treats $q_k$ as a dependency context anchor.
When $q_k$ is absent, each request to such $q_j$ incurs an extra miss attributable to the absence of $q_k$, so the incremental miss cost scales with the aggregate downstream request mass.
Since $\kw{dep}(q_k)=\sum_{(q_k,q_j)\in\mathcal{E}_s}\kw{freq}(q_j)$ explicitly aggregates this mass, the miss increase is monotone in $\kw{dep}(q_k)$.
A complete proof is provided in Appendix~\ref{app:dep_proof}.

To make the structure-induced miss cost usable in an online policy, we need an online mechanism to maintain the dependency-link set $\mathcal{E}_s$ and hence the structural term $\kw{dep}(\cdot)$ in Definition~\ref{def:tsi}.

\etitle{Lightweight dependency link detector.}
To maintain $\mathcal{E}_s$ online, we attach each incoming query $q_t$ to at most one 
\emph{dependency parent} $q_p$ within the current topic episode.
If such a parent exists, we add a directed link $(q_p,q_t)$ to $\mathcal{E}_s$.
This one-parent design enables constant-time updates of the structural term $\kw{dep}(\cdot)$ in Definition~\ref{def:tsi}. Intuitively, the dependency parent should satisfy the requirements that it is
(1) recent enough to reflect the local conversational context, and
(2) semantically related to $q_t$, so that $q_t$ is likely to reuse its information.
Accordingly, \textsc{DetectParent} selects the parent from a small set of \emph{recent resident} predecessors.
Specifically, it scans only cached candidates $q_k\in\C$ within a look-back window $t-k\le T$ and with $\kw{sim}(q_k,q_t)\ge \tau$.

We impose the look-back window $T$ and restrict candidates to resident entries for three reasons.
(1) Dependencies inside an episode are typically local in time, so very old queries are unlikely to supply the context that $q_t$ reuses.
(2) A bounded window limits the candidate set and keeps the per-request detection cost stable.
(3) The parent is meant to be an in-cache context anchor; if a candidate is not currently in $\C$, linking $q_t$ to it neither improves hit probability nor yields a maintainable anchor state for downstream queries.

Therefore, for each candidate $q_k \in \C$, we compute
\[
\kw{score}(k,t)
=\frac{1}{t-k}\cdot \kw{sim}(q_k,q_t).
\]
This score combines: (1) a recency discount $\frac{1}{t-k}$, and (2)
semantic similarity $\kw{sim}(q_k,q_t)$.
Finally, \textsc{DetectParent} returns the Top-$1$ candidate under $\kw{score}(k,t)$ and returns $\emptyset$ if no candidate survives the filters.
The result is cached as $\kw{par}(q_t)$ for future accesses.

\begin{algorithm}[t]
\caption{Constant-time $\kw{TSI}$ Update}
\label{alg:tsi-update-lite}
\KwIn{incoming request $q_t$; cache $\C$; current time step $t$; look-back window $T$; similarity threshold $\tau$; weight $\lambda$}
\KwData{entry states $\kw{freq}(\cdot)$, $\kw{dep}(\cdot)$, $\kw{TSI}(\cdot)$, and a cached parent pointer $\kw{par}(\cdot)$ for each cached query}
\KwOut{updated $\kw{TSI}$ scores for $q_t$ and its dependency parent $q_p$ (if any)}
\SetKwProg{Fn}{Procedure}{}{end}
\Fn{\textsc{UpdateTSI}($q_t, \C, t$)}{
  $\kw{freq}(q_t) \leftarrow \kw{freq}(q_t) + 1$\;
  $\kw{TSI}(q_t) \leftarrow \kw{freq}(q_t) + \lambda \cdot \kw{dep}(q_t)$\;

  \eIf{$\kw{par}(q_t) \neq \emptyset$}{
    $q_p \leftarrow \kw{par}(q_t)$\; $\kw{new} \leftarrow 0$\;
  }{
    $q_p \leftarrow \textsc{DetectParent}(q_t, \C, t, T, \tau)$\;
    $\kw{par}(q_t) \leftarrow q_p$\; $\kw{new} \leftarrow 1$\;
  }

  \If{$q_p \neq \emptyset$ \textbf{\upshape and} $q_p \in \C$}{
    \eIf{$\kw{new}=1$}{
      $\kw{dep}(q_p) \leftarrow \kw{dep}(q_p) + \kw{freq}(q_t)$\;
    }{
      $\kw{dep}(q_p) \leftarrow \kw{dep}(q_p) + 1$\;
    }
    $\kw{TSI}(q_p) \leftarrow \kw{freq}(q_p) + \lambda \cdot \kw{dep}(q_p)$\;
  }
  \KwRet{$\kw{TSI}(q_t)$ and $\kw{TSI}(q_p)$}\;
}
\end{algorithm}

\etitle{Online $\kw{TSI}$ maintenance.}
Equipped with the dependency detector and cached parents, each access to $q_t$ triggers a constant-time update cascade:
(1) increment $\kw{freq}(q_t)$ and recompute $\kw{TSI}(q_t)$;
(2) if $\kw{par}(q_t)=\emptyset$, run \textsc{DetectParent} to assign a parent (otherwise reuse the cached parent);
(3) if $q_p\in\C$, update $\kw{dep}(q_p)$ and $\kw{TSI}(q_p)$ using Algorithm~\ref{alg:tsi-update-lite}.

\section{Evaluation}
\label{sec:evaluation}

\subsection{Evaluation Questions}
\label{sec:eval_questions}
The goal of our evaluation is twofold: to validate RAC in scenarios where standard policies fail, and to verify the effectiveness of our relation-aware signals and unified heuristic.
We design our experiments to answer the following research questions:

\textbf{Q1 Robustness under sparse recurrence. (~\ref{sec:main_results})}
On workloads with long reuse distances and sparse local recurrence, does \emph{RAC} maintain stable gains by leveraging relation-aware signals, whereas recency-/frequency-driven policies become unreliable?

\textbf{Q2 Effectiveness on real-world traces.(~\ref{sec:main_results})}
On timestamp-continuous \LLM\ dialogue traces from real datasets, how does \emph{RAC} compare with state-of-the-art eviction policies under identical semantic hit semantics?

\textbf{Q3 Ablation on components.(~\ref{sec:rq3_ablation})}
Under identical cache budgets and hit semantics, what is the marginal contribution of \emph{(A) Topical Prevalence (TP)} and \emph{(B) Topic Structural Importance (TSI)} to RAC's performance gains?

\textbf{Q4 Parameter impact and robustness.(~\ref{sec:rq4_params})}
How do key hyperparameters in \emph{RAC} affect cache hit rate and latency, and how stable is \emph{RAC} across a wide parameter range?

\begin{figure*}[t]
  \centering

  \includegraphics[width=0.52\textwidth]{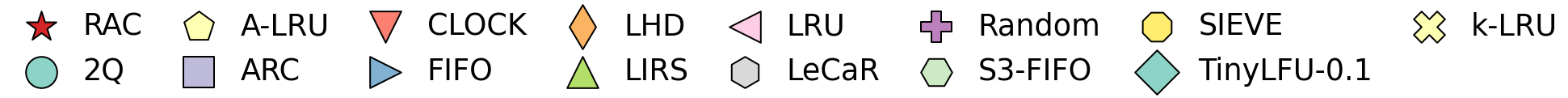}

  \vspace{0.3em}

  \begin{subfigure}[t]{0.49\textwidth}
    \centering
    \includegraphics[width=0.94\linewidth]{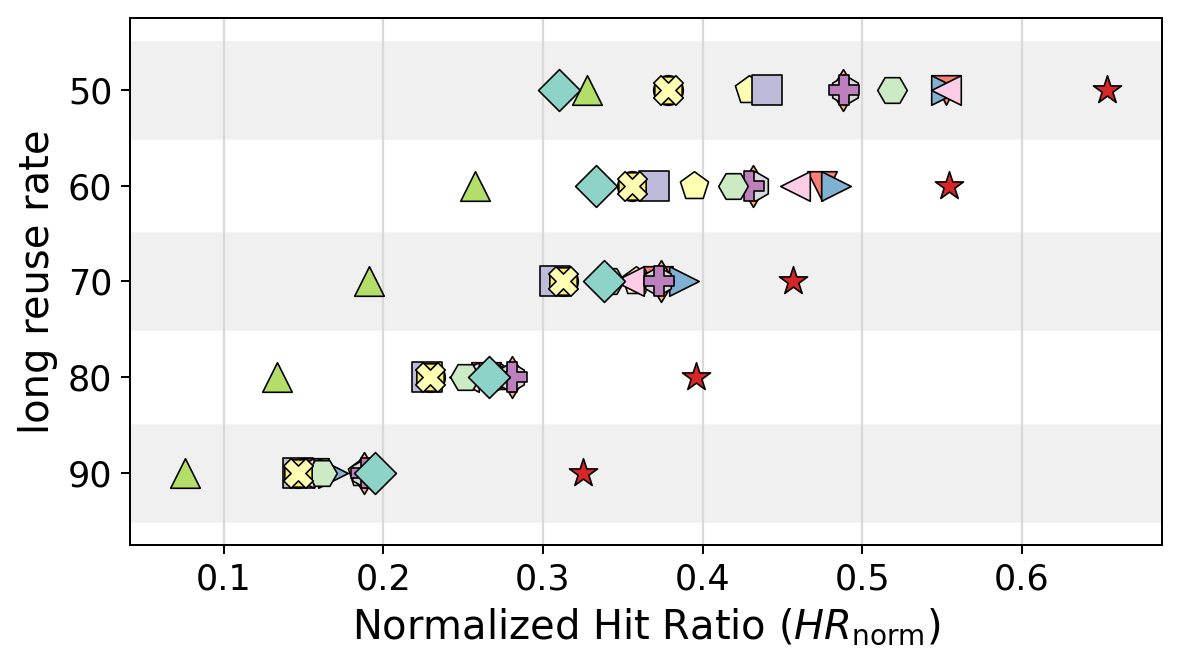}
    \caption{Simulated sequences varying long reuse-distance ratio}
    \label{fig:rq1_hr_two_axes_a}
  \end{subfigure}\hspace{0.6em}%
  \begin{subfigure}[t]{0.49\textwidth}
    \centering
    \includegraphics[width=0.94\linewidth]{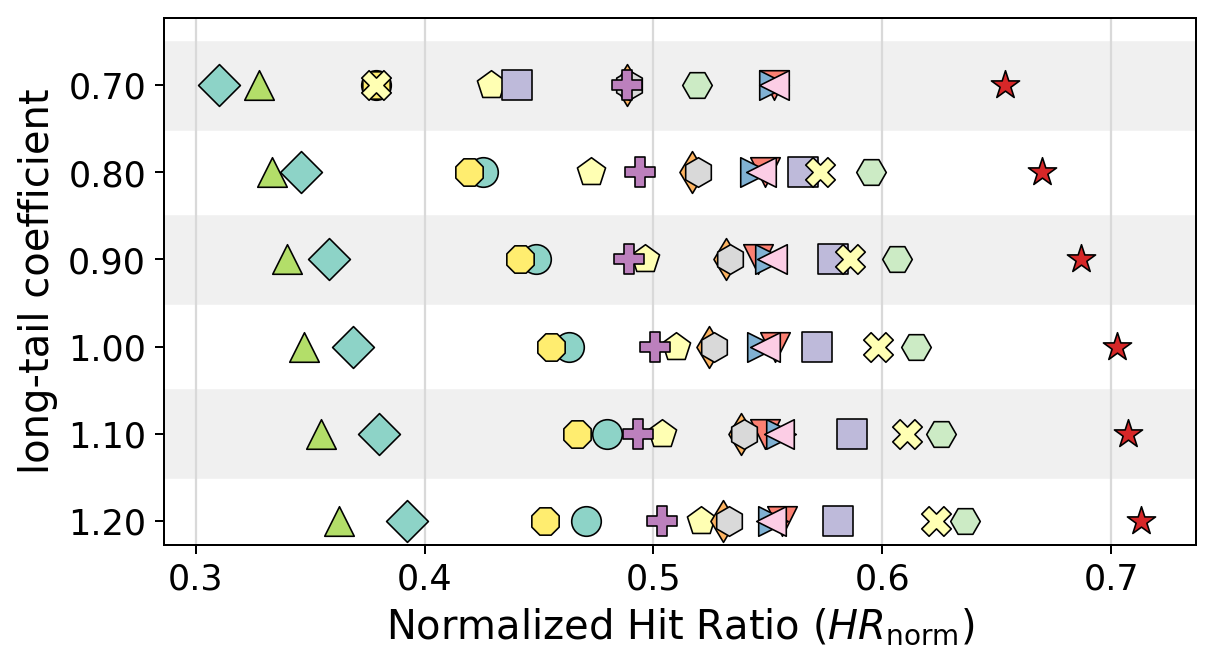}
    \caption{Simulated sequences varying long-tail coefficient}
    \label{fig:rq1_hr_two_axes_b}
  \end{subfigure}

  \vspace{0.2em}
  \caption{Hit ratio on simulated sequences under two stress axes: (a) varying long reuse-distance ratio; (b) varying long-tail coefficient.}
  \label{fig:rq1_hr_two_axes}
\end{figure*}

\begin{figure*}[t]
  \centering

  \begin{subfigure}[t]{0.06\textwidth}
    \centering
    \raisebox{0.18\height}{\includegraphics[width=0.50\linewidth]{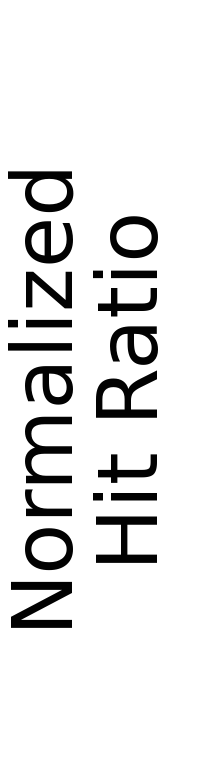}}
  \end{subfigure}\hspace{-0.15em}
  \begin{subfigure}[t]{0.305\textwidth}
    \centering
    \includegraphics[width=\linewidth]{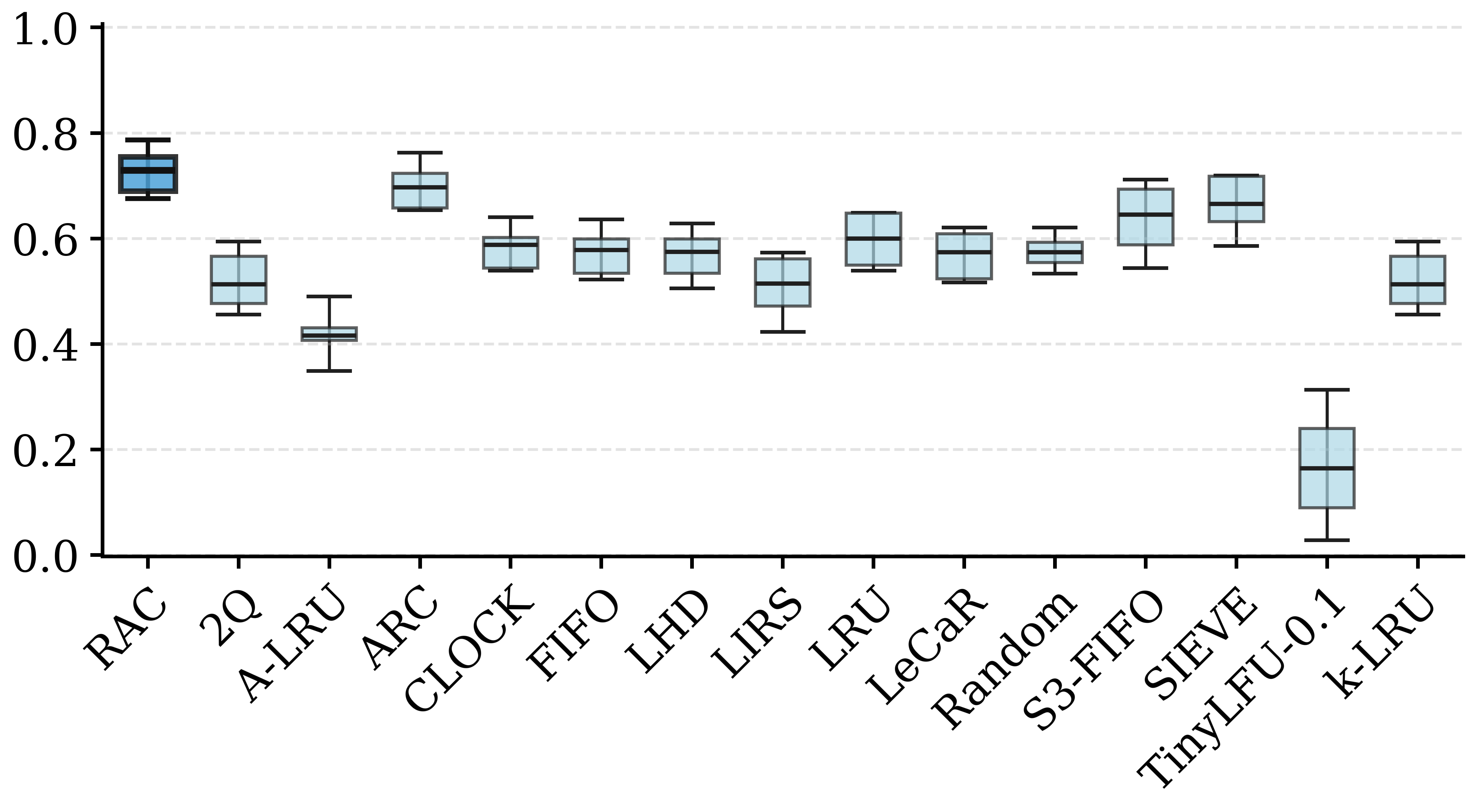}
    \caption{True trace at 2.5\% capacity}
    \label{fig:rq2_true_2p5}
  \end{subfigure}\hspace{0.25em}%
  \begin{subfigure}[t]{0.305\textwidth}
    \centering
    \includegraphics[width=\linewidth]{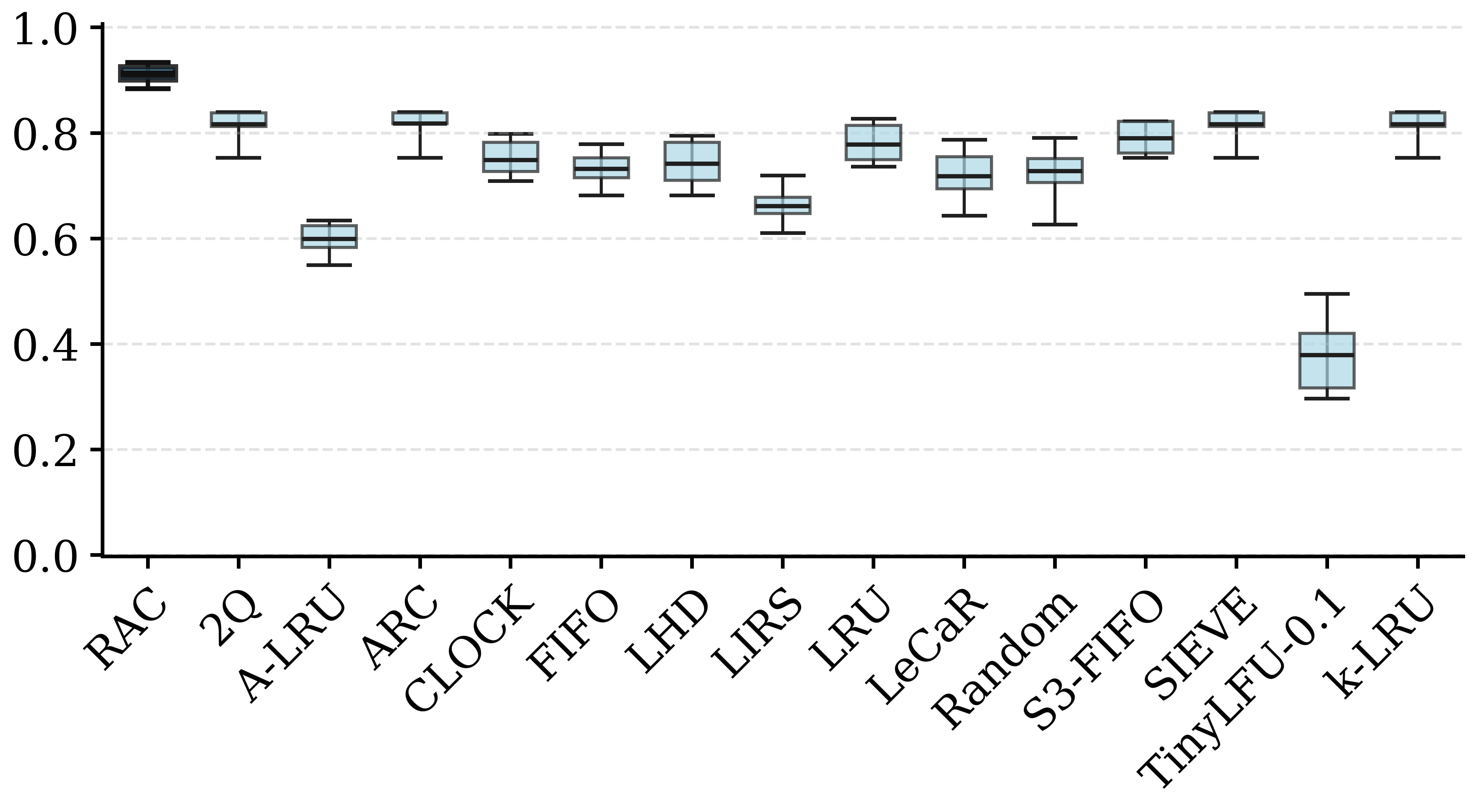}
    \caption{True trace at 10\% capacity}
    \label{fig:rq2_true_10}
  \end{subfigure}\hspace{0.25em}%
  \begin{subfigure}[t]{0.305\textwidth}
    \centering
    \includegraphics[width=\linewidth]{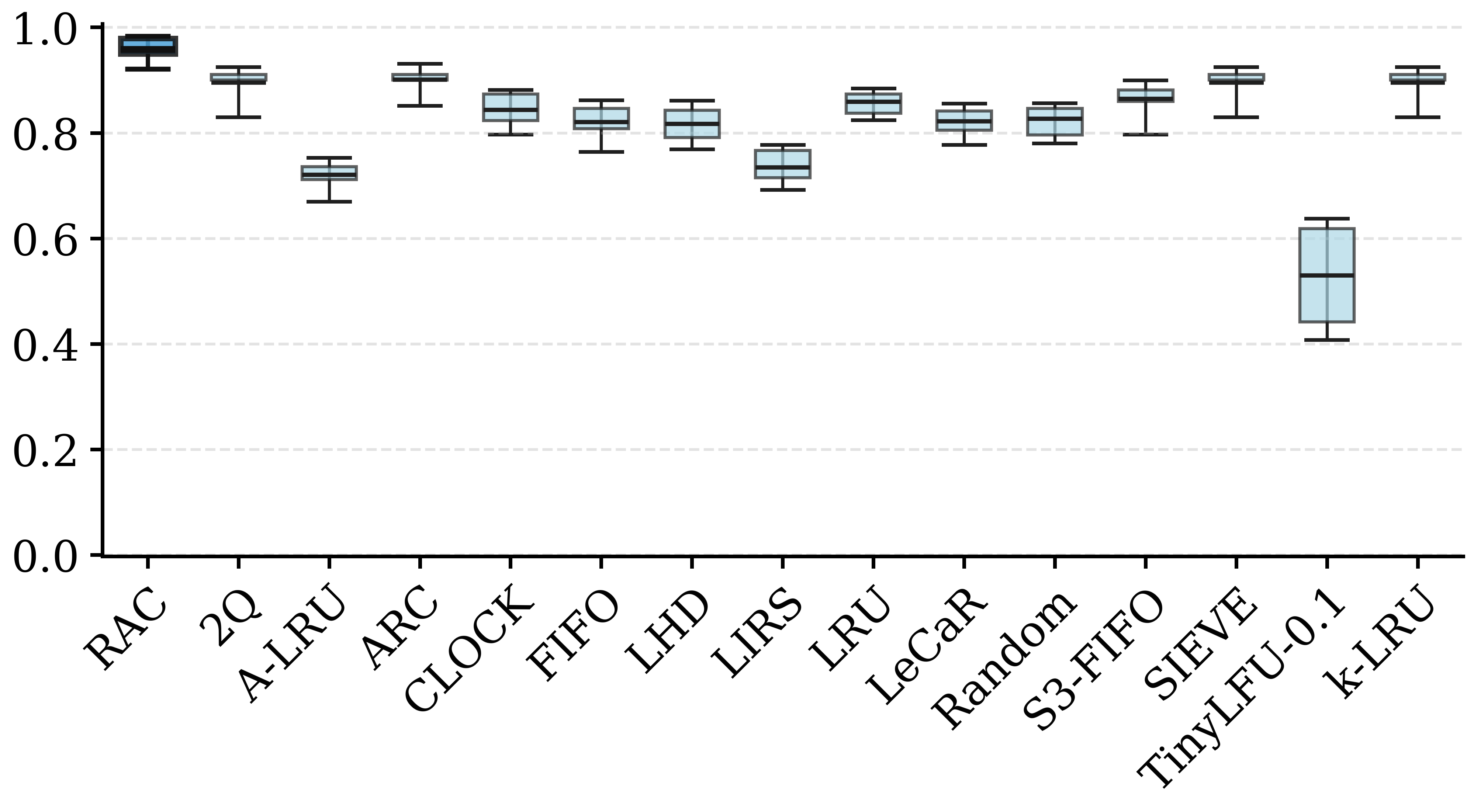}
    \caption{True trace at 20\% capacity}
    \label{fig:rq2_true_20}
  \end{subfigure}

  \vspace{-0.6em}
  \caption{Normalized hit ratio on timestamp-continuous OASST1 sub-traces under different cache capacities.}
  \label{fig:rq2_summary_realtraces_caps}
\end{figure*}

\subsection{Evaluation Setup}
\label{sec:eval_setup}

\stitle{Workloads.}
We utilize both real-world logs and synthetically constructed traces to evaluate performance across varying degrees of locality and contention.

\etitle{Real traces (OASST1).}
To answer \textbf{RQ2}, we extract traces from OASST1, a publicly available human-assistant dialogue corpus collected by the OpenAssistant project and released with conversation-thread structure and message-level metadata.
OASST1 provides chronological timestamps and conversation-thread structure, which enables constructing timestamp-continuous traces without splitting or interleaving threads.
We extract \textbf{10 non-overlapping, timestamp-continuous sub-traces}, each containing \textbf{10{,}000} dialogue requests.
All policies process the \emph{same} request sequence under identical semantic hit semantics for a fair comparison.

\etitle{Synthetic traces.}
To address \textbf{RQ1} by overcoming the limited structural diversity of real traces, we construct synthetic traces that, under tight cache budgets, induce workload regimes where recency and frequency become weak reuse predictors, and thus stress-test whether RAC sustains gains when short-window policies fail.
We adopt a topic-level semi-Markov generator: each trace concatenates variable-length topic episodes, where each episode is a complete multi-turn session that is never split or interleaved; hence topic switches occur only at session boundaries.

We sample $N=120$ topics and maintain a session pool, expanding each topic to $\sim$40 complete sessions (original + generated variants) using ChatGPT.
Within each topic, sessions exhibit context-ordered dependencies that can be abstracted as a DAG; we preserve this property in generated variants while enriching structural diversity by extending dependency branches with consistent continuations.

We fix cache capacity to $C{=}1000$ and trace length to $10{,}000$ requests.
We conduct controlled sweeps along two workload axes, generating \textbf{20 independent traces per setting}:
\begin{itemize}
  \item \emph{Reuse-distance axis.}
  We call a reuse \emph{long} if its reuse distance exceeds the cache capacity $C$.
  The \emph{long-reuse ratio} is the fraction of reuse events in a trace that are long under this definition.
  Fixing $\gamma{=}0.7$, we vary the long-reuse ratio from 50\% to 90\% (step 10\%) by repeating prior sessions and placing repeats at randomized positions.

  \item \emph{Long-tail skew axis.}
  We fix the long-reuse ratio at 50\% and vary the Zipf exponent $\gamma\in\{0.7,0.8,\ldots,1.2\}$, which controls how concentrated the topic popularity is (smaller $\gamma$ yields a flatter distribution, while larger $\gamma$ yields a heavier head and a longer tail).
  This range is consistent with empirical observations that query frequencies in real logs follow heavy-tailed Zipf/power-law-like distributions~\cite{petersen2016powerlaw,lillo2021estimating}.
\end{itemize}

\stitle{Metrics(Normalized Hit Ratio $HR_{\text{norm}}$).}
Let $HR_{\text{algo}}(C)$ be the hit ratio of an eviction policy under cache capacity $C$ on the given request sequence (with the same hit semantics).
Let $HR_{\text{full}}$ be the hit ratio of an infinite cache on the same sequence, i.e., the upper bound when no evictions occur.
We report
\[
HR_{\text{norm}}(C)\;=\;\frac{HR_{\text{algo}}(C)}{HR_{\text{full}}}.
\]
This normalization controls for dataset-dependent hit ceilings and highlights the effect of the eviction policy.

\stitle{Capacity Configuration.}
We report cache capacity as a fraction of the total unique request footprint in each trace (e.g., $C\in\{50,\ldots,400\}$ for OASST1), covering the cache-cliff to moderately provisioned regimes.
For \textbf{RQ1} (synthetic), we fix capacity at \textbf{10\%}; for \textbf{RQ2} (real), we evaluate \textbf{2.5\%}, \textbf{10\%}, and \textbf{20\%}; and for \textbf{RQ3} (ablation), we sweep \textbf{2.5\%--20\%} with a \textbf{2.5\%} step.

\stitle{Methods and baselines.}
We evaluate RAC against representative eviction policies under identical hit semantics and cache budgets.

\begin{itemize}
  \item \textbf{Our methods:} \textbf{RAC} (full; TP+TSI), \textbf{RAC w/o TP} (TSI only), and \textbf{RAC w/o TSI} (TP only). The ablations isolate each component's marginal contribution (\textbf{RQ3}).
  \item \textbf{External baselines:}
  \begin{itemize}
    \item \textbf{Classic heuristics:} \texttt{FIFO}, \texttt{LRU}, \texttt{CLOCK}, \texttt{TTL}. \texttt{LRU} is a common production default (including KV caches).
    \item \textbf{Frequency-/scan-resistant:} \texttt{TinyLFU}, \texttt{ARC}, \texttt{S3-FIFO}, \texttt{SIEVE}, \texttt{2Q}. These policies mitigate scans but do not model semantic relations.
    \item \textbf{Learning-based:} \texttt{LHD}, \texttt{LeCaR}, which adapt online using observed access outcomes.
  \end{itemize}
\end{itemize}

\stitle{Implementation Details.}
\stitle{Implementation Details.}
Unless otherwise specified, RAC uses Top-1 vector retrieval with a topic-routing (hit) threshold $\tau=0.85$, a conservative value calibrated to match ChatGPT-judged semantic equivalence.
We maintain relation evidence with an edge-pruning threshold $\tau_{\text{edge}}=0.6$ and decay factor $\lambda=1$.
For synthetic workloads, unless a stress axis is being swept, we fix the Zipf exponent to $\gamma=0.7$ and set the long-reuse fraction to 50\%, where an event is counted as \emph{long reuse} if its reuse distance exceeds the cache capacity $C$.
To answer \textbf{RQ4}, we further vary these hyperparameters and report their sensitivity in \S\ref{sec:rq4_params}.

\begin{figure}[t]
  \centering
  \begin{minipage}[t]{0.5\linewidth}
    \centering
    \includegraphics[width=\linewidth]{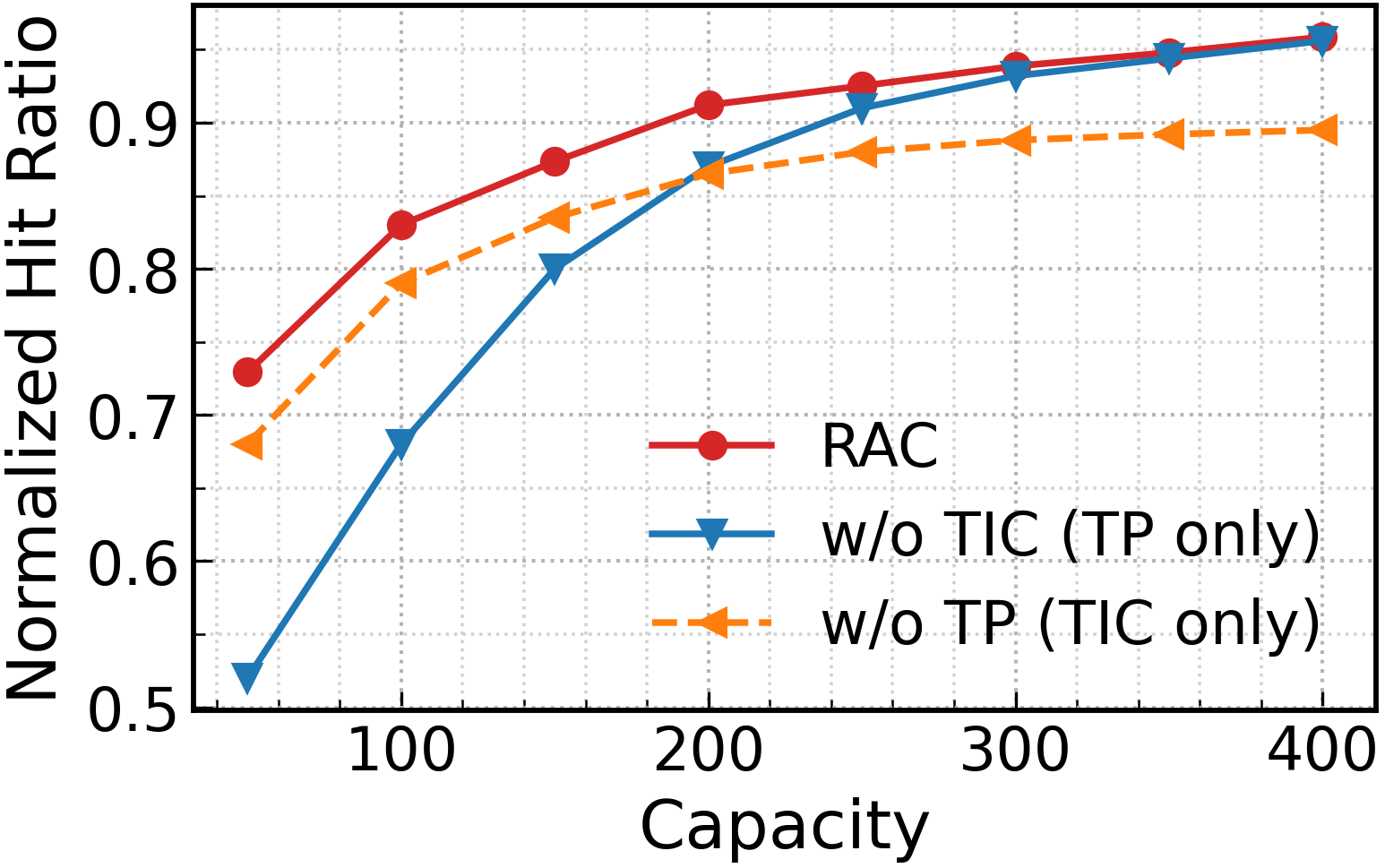}
    \caption*{\small (a) Performance vs. capacity}
  \end{minipage}\hspace{0em}%
  \begin{minipage}[t]{0.5\linewidth}
    \centering
    \includegraphics[width=\linewidth]{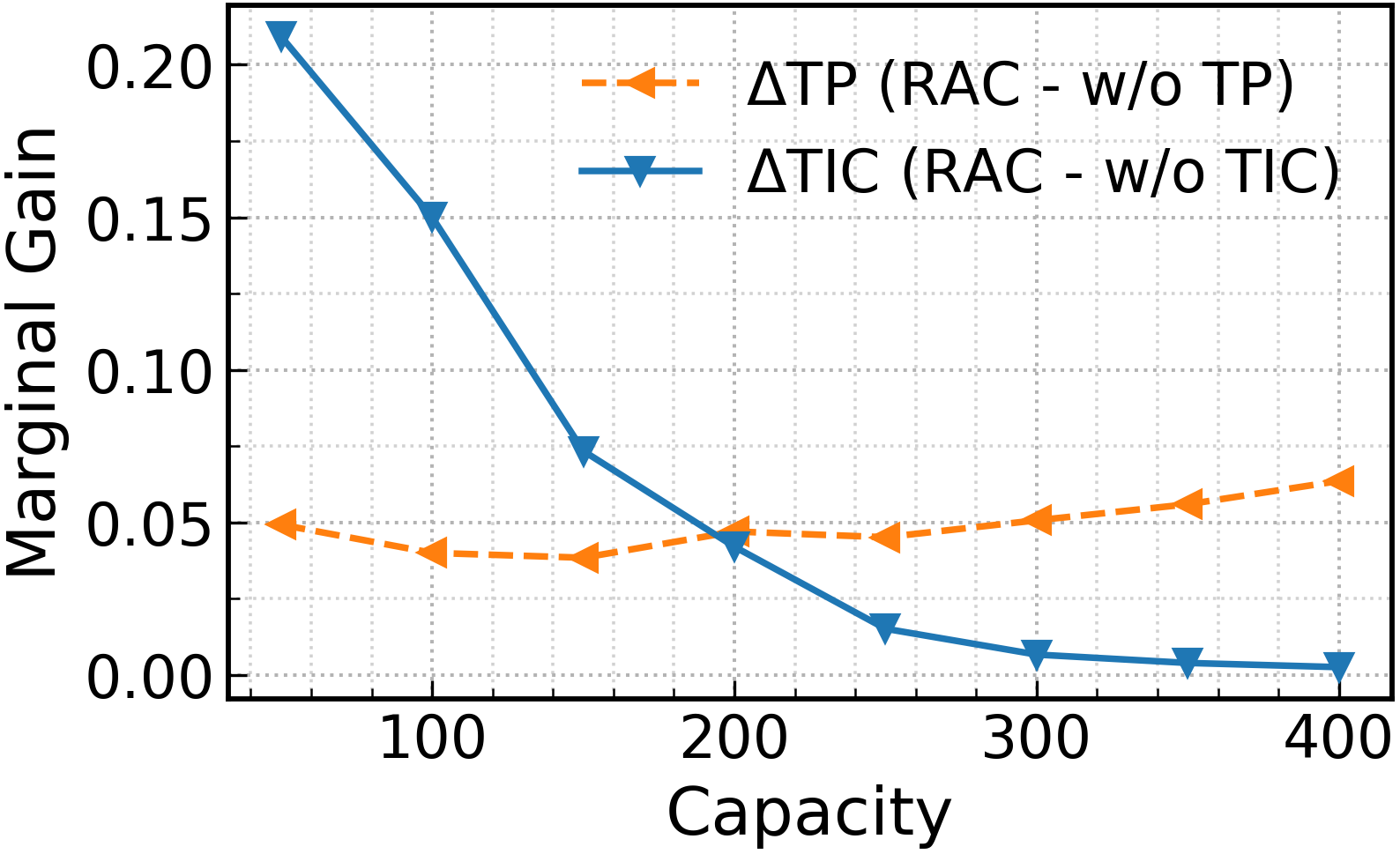}
    \caption*{\small (b) Marginal gains ($\Delta$TP, $\Delta$TIC)}
  \end{minipage}
  \caption{Ablation results under varying cache capacities.}
  \label{fig:rq3_ablation_two}
\end{figure}

\subsection{Main Results}
\label{sec:main_results}

We evaluate RAC under the unified simulator and the semantic hit semantics defined in \S\ref{sec:eval_setup}.
For \textbf{RQ1}, we conduct controlled stress tests on synthetic traces along the reuse-distance and popularity-skew dimensions (Figure~\ref{fig:rq1_hr_two_axes}).
For \textbf{RQ2}, we evaluate on timestamp-continuous OASST1 sub-traces under multiple cache capacities (Figure~\ref{fig:rq2_summary_realtraces_caps}).

\noindent\textbf{RQ1 (Robustness under sparse recurrence).}
Figure~\ref{fig:rq1_hr_two_axes} presents the hit ratios on simulated traces.
\emph{Impact of Reuse Distance (a):} RAC consistently ranks first, and its advantage widens as the workload becomes increasingly adversarial to short-window statistics.
As reuse shifts beyond the cache horizon, the recency and frequency signals relied upon by conventional policies weaken substantially.
Consequently, baselines deteriorate, whereas RAC degrades much more slowly, yielding a widening performance gap under stronger contention.
Quantitatively, compared with the strongest baseline, RAC achieves \emph{double-digit} relative gains under mild contention, growing to \emph{tens of percent} in the most challenging regimes.
Relative to the baseline average, these gains are even more pronounced.
These trends indicate that RAC remains robust precisely when sparse recurrence renders local hit statistics unreliable.

\emph{Impact of Popularity Skew (b):} While all methods generally improve as popularity concentrates on a smaller head set of topics, RAC maintains superiority across the entire tail-to-head spectrum.
RAC outperforms the strongest baseline by roughly \emph{15\%--20\%} consistently, and exceeds the baseline average by an even larger margin (typically \emph{tens of percent}).
This suggests that RAC not only benefits from frequency concentration but also extracts additional value within head topics by preserving structurally critical context anchors (TIC), thereby capturing semantic hits that queue- or sketch-based policies miss.

\noindent\textbf{RQ2 (Effectiveness on real-world traces).}
Figure~\ref{fig:rq2_summary_realtraces_caps} reports performance on timestamp-continuous OASST1 sub-traces, spanning capacities from the cache-cliff to moderately provisioned regimes.
RAC achieves the best overall effectiveness under identical hit semantics across all evaluated capacities.
Specifically, RAC improves over the strongest baseline by \emph{around 5\%--12\%}, and over the baseline average by a larger margin (often \emph{well above 10\%}, especially under tighter budgets).
These results confirm that RAC's relation-aware design generalizes beyond controlled simulations to real dialogue workloads, where irregular topic shifts and heavy-tailed reuse distances often invalidate short-window statistics.
By aggregating topic-level evidence (TP) and prioritizing structurally critical context anchors (TIC), RAC delivers stable improvements without altering the hit predicate.

\subsection{Ablation Study (RQ3)}
\label{sec:rq3_ablation}

To address \textbf{RQ3}, we isolate the contributions of RAC's two relation-aware signals by comparing the full model with two ablated variants:
(i) \emph{RAC w/o TP}, which removes topic-level evidence aggregation (topical prevalence);
and (ii) \emph{RAC w/o TIC}, which removes intra-topic dependency aggregation.
Figure~\ref{fig:rq3_ablation_two} reports normalized hit ratios and the corresponding marginal gains.

Figure~\ref{fig:rq3_ablation_two}(a) shows that full RAC consistently outperforms both ablations across all cache sizes, indicating that TP and TIC are complementary: removing either component yields substantial degradation.
Figure~\ref{fig:rq3_ablation_two}(b) further reveals their distinct roles across the capacity spectrum.
In the \emph{cache-cliff regime} (tight budgets), removing TIC causes the sharpest drop, highlighting the importance of preserving within-topic context anchors when capacity is scarce.
As capacity increases, the marginal benefit of TIC diminishes naturally because larger caches retain such anchors by default.

In contrast, TP contributes persistently across all capacities.
Removing TP induces steady degradation even at larger budgets, suggesting that topic-level aggregation remains crucial for long-horizon revisits and irregular topic transitions where local statistics remain sparse.
Overall, RAC's gains arise from the synergy between TP (long-horizon topical awareness) and TIC (intra-topic discrimination), ensuring robust effectiveness across diverse budgets.

\begin{figure}[t]
  \centering
  \begin{subfigure}[t]{0.06\columnwidth}
    \centering
    \includegraphics[width=\linewidth]{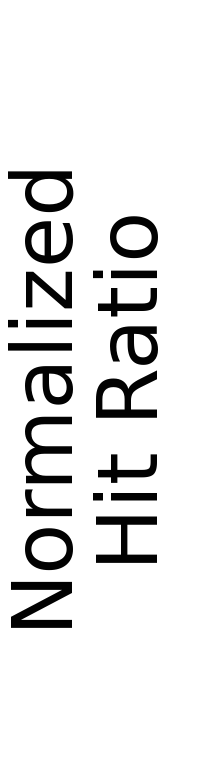}
    \caption*{}
  \end{subfigure}\hfill
  \begin{subfigure}[t]{0.30\columnwidth}
    \centering
    \captionsetup{width=1.2\linewidth}
    \includegraphics[width=\linewidth]{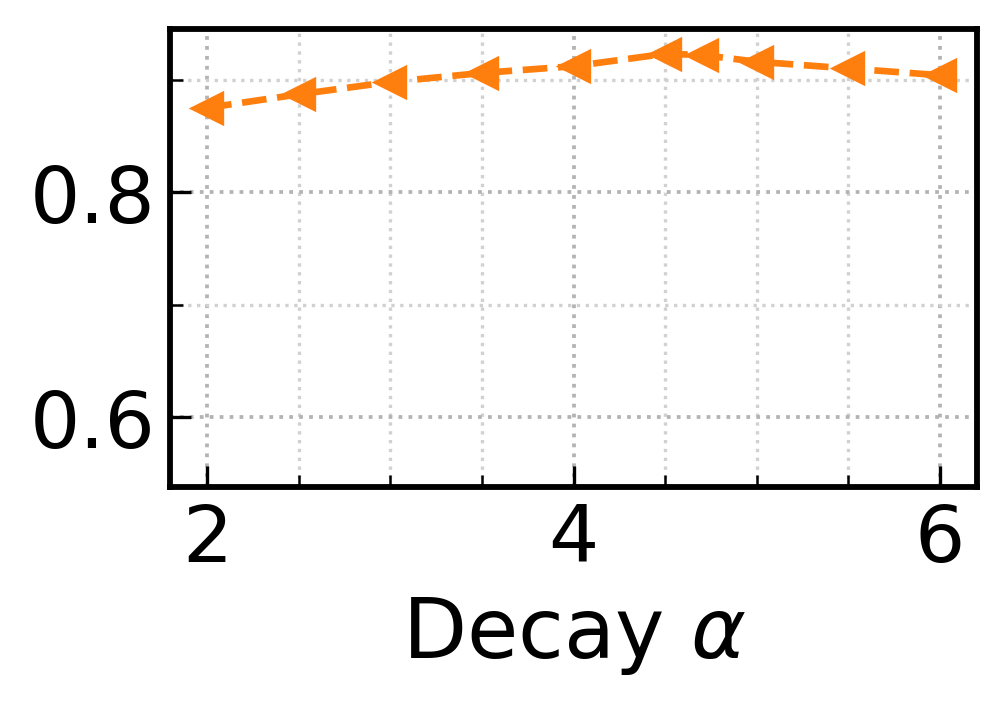}
    \caption{Decay coefficient~$\alpha$}
    \label{fig:rq4-alpha}
  \end{subfigure}\hfill
  \begin{subfigure}[t]{0.30\columnwidth}
    \centering
    \captionsetup{width=1.2\linewidth}
    \includegraphics[width=\linewidth]{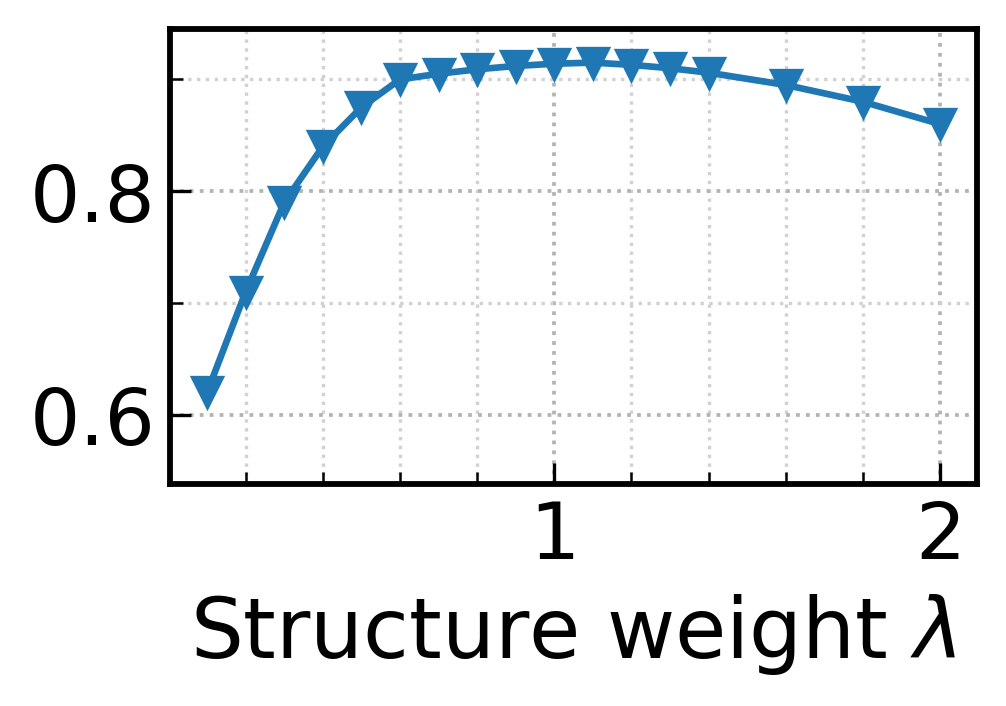}
    \caption{Weight~$\lambda$}
    \label{fig:rq4-lambda}
  \end{subfigure}\hfill
  \begin{subfigure}[t]{0.30\columnwidth}
    \centering
    \captionsetup{width=1.2\linewidth}
    \includegraphics[width=\linewidth]{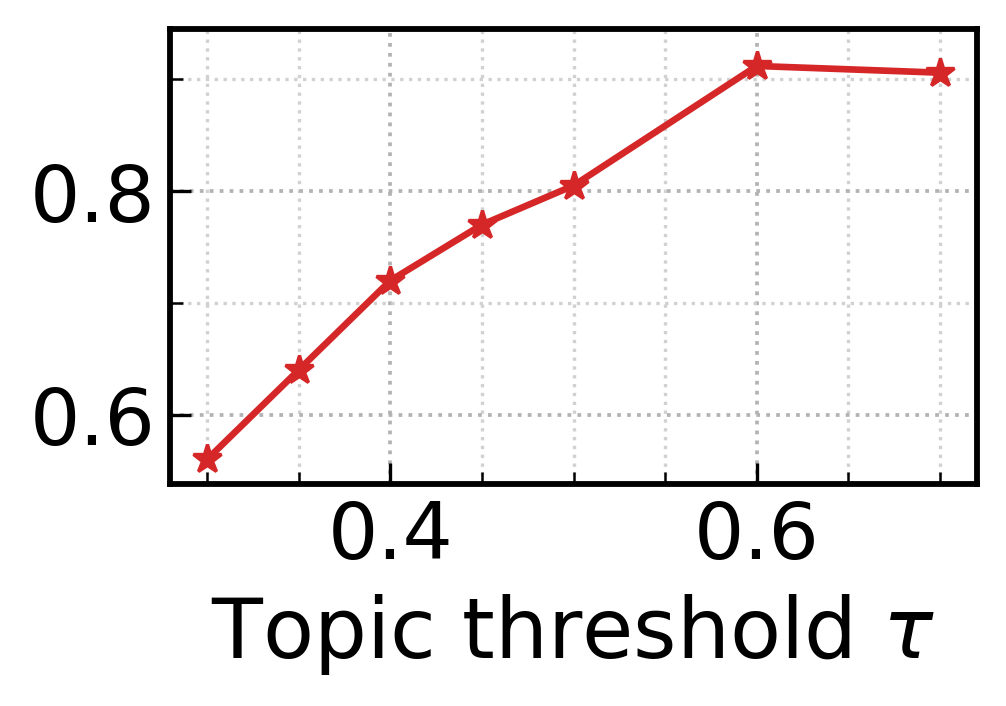}
    \caption{Threshold~$\tau$}
    \label{fig:rq4-tau}
  \end{subfigure}
  \caption{RQ4: Parameter sensitivity at 10\% cache capacity.}
  \label{fig:rq4-params}
\end{figure}

\subsection{Parameter Sensitivity (RQ4)}
\label{sec:rq4_params}

To answer \textbf{RQ4}, Figure~\ref{fig:rq4-params} reports RAC's sensitivity to three key hyperparameters at 10\% cache capacity: routing threshold $\tau$, decay factor $\alpha$, and structural weight $\lambda$.
Overall, RAC exhibits a broad region of stable performance: once a parameter enters a reasonable operating range, performance becomes relatively insensitive to moderate changes.
This indicates that RAC does not rely on delicate tuning to achieve strong effectiveness.

\noindent\textbf{Routing threshold $\tau$.}
As shown in Figure~\ref{fig:rq4-params} (right), an overly small $\tau$ makes routing permissive, admitting noisy matches that dilute effective reuse.
Increasing $\tau$ into a reasonable semantic-matching range improves effectiveness substantially before saturating; further increasing $\tau$ yields limited gains and may slightly reduce reuse by becoming overly strict.
This supports the use of a strict semantic gate while demonstrating stability within a sensible band.

\noindent\textbf{Decay factor $\alpha$.}
Regarding temporal aggregation (Figure~\ref{fig:rq4-params}, left), when $\alpha$ is too small, topic evidence decays slowly and can overweight stale history; increasing $\alpha$ improves effectiveness by discounting older evidence appropriately.
However, when $\alpha$ becomes overly large, evidence is discounted too aggressively, making decisions more reactive to short-term fluctuations and slightly degrading performance.
The curve is smooth with a wide near-optimal region, indicating robustness to $\alpha$.

\noindent\textbf{Structural weight $\lambda$.}
Regarding the TP--TIC trade-off (Figure~\ref{fig:rq4-params}, middle), small $\lambda$ underutilizes dependency signals, whereas excessively large $\lambda$ overemphasizes structure at the expense of topical prevalence.
The best performance occurs in a moderate-to-high band and remains close to optimal over a broad interval, consistent with our default setting.
Taken together, these results show that each parameter has an interpretable effect and that RAC maintains stable performance over wide ranges.

\subsection{Evaluation Summary}
\label{sec:eval_summary}
\begin{itemize}
  \item \textbf{For Q1 (Robustness under sparse recurrence):}
  RAC is robust when local hit statistics become uninformative under sparse recurrence, sustaining stable advantages by relying on relation-aware evidence rather than short-window recency/frequency signals.

  \item \textbf{For Q2 (Effectiveness on real-world traces):}
  RAC generalizes to timestamp-continuous real dialogue traces and remains effective under identical hit semantics, indicating that its gains are not an artifact of synthetic construction.

  \item \textbf{For Q3 (Ablation on components):}
  TP and TSI play complementary roles: TP captures reuse regularities via topic-level aggregation (beyond independent items), while TSI enables efficient within-topic discrimination of reuse value, which stabilizes performance under heavy-tailed popularity.

  \item \textbf{For Q4 (Parameter impact and robustness):}
  RAC is insensitive to moderate hyperparameter perturbations within a wide operating region, suggesting that its effectiveness does not depend on delicate tuning.
\end{itemize}

\section{Related Work}
\label{sec:related-work}

We review prior work from two perspectives: reusing computed content in \LLM serving, and cache eviction under limited capacity.

\smallskip
\stitle{Reuse of computed content in \LLM serving.}
Focusing on caching inference artifacts to skip recomputation, our method supports both \emph{semantic matching} and \emph{compositional chunking}.

\etitle{Semantic caches.}
Focusing on reusing \emph{final generations} via similarity, GPTCache employs embedding lookup \cite{bang2023gptcache}, MeanCache leverages cluster representatives \cite{gillmeancache}, and ScaLM optimizes for scalable serving constraints \cite{li2024scalm}.

\etitle{Compositional \kw{Chunk KV} reuse.} 
Targeting \emph{intermediate \kw{Chunk KV}} reuse, PromptCache and CacheBlend enable modular reuse via structural compatibility \cite{gim2024promptcache} and component blending \cite{yao2025cacheblend}. RAGCache and Cache-Craft optimize management via hierarchical policies \cite{Jin2024RAGCache} and composability-based retention \cite{Agarwal2025CacheCraft}, respectively, while KVCache in the Wild characterizes production-level memory pressure \cite{wang2025kvcachewild}.

\smallskip
\stitle{Eviction algorithms under limited capacity.}
Eviction policies use lightweight metadata to prioritize cached items and select victims under capacity pressure.

\etitle{Simple heuristics.} These policies use deterministic metrics for ranking. Belady's MIN establishes the offline optimum \cite{Belady1966Optimal,Mattson1970Stack}, while LRU/FIFO implement constant-overhead tail eviction \cite{Corbato1969Multics,Sleator1985Competitive}. Similarity caching groups items via modified hit semantics \cite{Neglia2022SimilarityCaching, Sabnis2021GRADES}. GDS/GDSF and LRFU unify size, cost, and frequency into tunable priority scores \cite{Cao1997GDS,Cherkasova1998GDSF,lee2001lrfu}.

\etitle{Composite-structure heuristics.} These designs employ multi-level queues or auxiliary structures to refine value estimation. 2Q, ARC/CAR, and TinyLFU filter one-timers and churn via probationary queues or sketches \cite{Johnson1994TwoQ,megiddo2004arc,Einziger2017TinyLFU}. S3-FIFO and Sieve offer high-performance demotion rules with minimal metadata \cite{liu2023s3fifo,Zhang2023Sieve}. LRU-$k$ and LIRS/LIRS2 utilize reuse distance to separate persistent items from transients \cite{ONeil1993LRUk,Jiang2002LIRS,Zhong2021LIRS2}. In \LLM runtimes, SGLang manages \kw{KV} states via radix trees; however, its eviction remains LRU-like and is primarily constrained by parent--child dependencies to preserve structural validity \cite{zheng2024sglang}.

\etitle{Learning-based eviction.} Learning-guided policies approximate utility via data-driven signals. Hawkeye and Glider learn cache-friendliness from OPT-derived supervision \cite{jain2016hawkeye,shi2019glider}. LRB/HALP and ILCache map features to priorities approximating reuse distance \cite{song2020lrbelady,song2023halp,liu2020ilcache}, while LeCaR and CACHEUS adaptively weight experts via online feedback \cite{Hashemi2018LeCaR,Rodriguez2021CACHEUS}. DRLCache and its variants further model replacement as a sequential process via reinforcement learning \cite{zhong2017drlcache,zhou2022endtoend}.

This work differs from previous approaches in three key aspects:
(1) unlike \textbf{simple heuristics}, we estimate value from semantic associations and structural dependencies rather than isolated hit statistics;
(2) unlike \textbf{composite heuristics}, we use query relations to identify and preserve critical context anchors instead of auxiliary admission/filtering structures;
and (3) unlike prior \textbf{learning-based policies}, we adopt a lightweight, interpretable online rule that adapts to structural shifts without training or inference overhead.

\section{Conclusion}
\label{sec:conclusion}

We present Relation-Aware Cache (\RAC), an online cache replacement strategy tailored for dynamic \LLM workloads. Based on the insight that reuse is governed by discourse structure, \RAC substitutes isolated item signals with more stable topic-level frequency and recency, and utilizes Structural Importance to identify critical dependency anchors, thereby minimizing miss rates caused by sparse recurrence. We evaluate \RAC on both real-world dialogue traces and synthetic benchmarks. Experimental results demonstrate that \RAC outperforms state-of-the-art baselines by 20\% in hit ratio, validating that exploiting relation-aware signals is key to efficient memory management in long-context serving.

\clearpage
\balance
\bibliographystyle{ACM-Reference-Format}
\bibliography{paper}

\clearpage
\section{Complete Proof for Theorem~\ref{thm:dep_miss} and Extensions}
\label{app:dep_proof}

This appendix provides the complete proof for Theorem~\ref{thm:dep_miss} in the main text, which formalizes the monotone relation between the eviction-induced miss increase and the downstream dependency mass $\kw{dep}(\cdot)$.
We then outline an extension that connects dependency-DAG structure to a principled notion of structural-importance ranking.

\subsection{Complete Proof of Theorem~\ref{thm:dep_miss}}
\label{app:dep_proof:full}

\etitle{Setup and prerequisite semantics.}
Fix a topic $s$ and its dependency links $\mathcal{E}_s$.
Consider the in-topic request sequence $\{Q_t\}_{t=1}^{T}$, where $Q_t$ denotes the realized query at time $t$.
We interpret each edge $(q_k,q_j)\in\mathcal{E}_s$ as a prerequisite relation: $q_k$ is a context anchor required by $q_j$ within the topic.
Accordingly, if $q_k$ is absent from the cache when serving a request to such a dependent query $q_j$, then the request to $q_j$ incurs \emph{at least one additional miss} that is attributable to the absence of $q_k$ (independent of how other entries are arranged).
This defines a conservative, \emph{unavoidable} miss semantics: we only charge misses that must occur due to the missing prerequisite anchor.

\etitle{Unavoidable miss lower bound induced by evicting $q_k$.}
Let $\mathcal{N}(q_k)\triangleq\{q_j:(q_k,q_j)\in\mathcal{E}_s\}$ denote the set of one-hop dependents of $q_k$.
We define, over a horizon $T$, the number of \emph{unavoidable additional misses} induced by the absence of $q_k$ as
\begin{equation}
\label{eq:deltaT_def}
\Delta_T(q_k)
\;\triangleq\;
\sum_{t=1}^{T}\mathbbm{1}\{Q_t\in \mathcal{N}(q_k)\}.
\end{equation}
That is, $\Delta_T(q_k)$ counts how many requests in the horizon target dependent queries of $q_k$.
Under the prerequisite semantics, if $q_k$ is evicted (and thus missing whenever these dependent requests arrive), then each of these requests must incur at least one extra miss attributable to the missing anchor.

\begin{proof}[Proof of Theorem~\ref{thm:dep_miss}]
Fix topic $s$ and an anchor query $q_k$.
Consider any time step $t\in\{1,\ldots,T\}$.

If $Q_t\notin\mathcal{N}(q_k)$, then by definition there is no edge $(q_k,Q_t)\in\mathcal{E}_s$, and the prerequisite semantics imposes no additional miss that is attributable to the absence of $q_k$ for serving $Q_t$.
Thus, the absence of $q_k$ does not \emph{force} an extra miss at time $t$.

If $Q_t\in\mathcal{N}(q_k)$, then there exists a dependent query $q_j$ such that $Q_t=q_j$ and $(q_k,q_j)\in\mathcal{E}_s$.
By the prerequisite semantics, whenever $q_k$ is absent, serving this request to $q_j$ must incur \emph{at least one} additional miss attributable to the missing anchor $q_k$.
Therefore, for each such time step $t$, the eviction of $q_k$ contributes at least one unit to the eviction-induced miss increase.

Summing these unavoidable contributions over $t=1$ to $T$, we conclude that over any horizon $T$, the cumulative miss increase attributable to evicting $q_k$ is lower bounded by the number of dependent-query requests within the horizon, namely $\Delta_T(q_k)$ in~\eqref{eq:deltaT_def}.
Equivalently, the eviction-induced miss increase is monotonically increasing in $\Delta_T(q_k)$.

We now connect this lower bound to the online statistic $\kw{dep}(q_k)$.
Recall the definition
\[
\kw{dep}(q_k)\;\triangleq\;\sum_{(q_k,q_j)\in\mathcal{E}_s}\kw{freq}(q_j),
\]
where $\kw{freq}(q_j)$ is the observed hit/request mass of $q_j$ within topic $s$ over the same measurement horizon.
Since $\kw{dep}(q_k)$ aggregates the downstream mass over exactly the dependent set $\mathcal{N}(q_k)$, it is a direct empirical proxy for how frequently requests fall into $\mathcal{N}(q_k)$, i.e., for $\Delta_T(q_k)$ up to a common scaling determined by the horizon and counting convention.
Hence, the eviction-induced miss increase (via its unavoidable lower bound) is monotonically increasing in $\kw{dep}(q_k)$.
\end{proof}

\subsection{Extension: Structural-Importance Ranking on a Dependency DAG}
\label{app:dep_proof:dag}

Theorem~\ref{thm:dep_miss} motivates a first-order proxy that aggregates \emph{direct} dependents.
Within a topic, prerequisite relations form a directed acyclic graph (DAG), and the eviction externality of an anchor can be shaped by the \emph{global} dependency structure rather than only its one-hop neighborhood.
To obtain a lightweight structural ordering beyond one-hop counts, we adopt a PageRank/TextRank-style graph-based ranking, which assigns vertex importance via the stationary distribution of a random walk on a directed graph~\cite{page1999pagerank,mihalcea2004textrank,erkan2004lexrank}.

\etitle{Reverse-edge construction.}
Fix a topic and let $G=(V,E)$ denote its prerequisite DAG, where each node $v\in V$ is a query (cache entry) in the topic, and each directed edge $(u\!\to\! v)\in E$ indicates that $u$ is a prerequisite context anchor for $v$.
To rank anchors, we propagate importance from dependents back to prerequisites.
Concretely, we run the ranking walk on the reversed edges: for each prerequisite link $(u\!\to\! v)\in E$, we use a reversed link $(v\!\to\! u)$ in the walk.

\etitle{Random walk with uniform restart.}
Let $\mathrm{Out}(v)\triangleq\{u\in V:\ (v\!\to\! u)\in E\}$ denote the outgoing neighbors of $v$ in the reversed graph.
We define a random walk over $V$ with damping parameter $\beta\in(0,1)$:
with probability $\beta$, the walk follows a reversed edge chosen uniformly from $\mathrm{Out}(v)$; with probability $1-\beta$, it restarts to a uniformly random node in $V$.
For dangling nodes ($|\mathrm{Out}(v)|=0$), we fall back to a uniform jump.
Formally, the transition probability is
\[
P(u\mid v)\;\triangleq\;
\begin{cases}
\frac{1}{|\mathrm{Out}(v)|}, & u\in \mathrm{Out}(v)\ \text{and}\ |\mathrm{Out}(v)|>0,\\[4pt]
\frac{1}{|V|}, & |\mathrm{Out}(v)|=0,\\[4pt]
0, & \text{otherwise}.
\end{cases}
\]
We define the structural-importance score $r:V\to\mathbb{R}_{\ge 0}$ as the stationary distribution of this Markov chain, i.e., the unique solution to
\begin{equation}
\label{eq:rw_restart_uniform}
r(u)\;=\;(1-\beta)\cdot\frac{1}{|V|}
\;+\;\beta\sum_{v:\ (v\to u)\in E} P(u\mid v)\,r(v).
\end{equation}

\begin{proposition}[Existence, uniqueness, and computability]
\label{prop:exist_unique_uniform}
For any $\beta\in(0,1)$, the uniform-restart walk is irreducible and aperiodic.
Hence the stationary distribution in~\eqref{eq:rw_restart_uniform} exists and is unique.
Moreover, it can be computed by power iteration: starting from any distribution $r^{(0)}$ over $V$, repeatedly applying the update in~\eqref{eq:rw_restart_uniform} yields a sequence $\{r^{(t)}\}$ that converges to $r$.
\end{proposition}

\etitle{Centrality induced by dependency structure.}
The score $r(u)$ is the long-run visit probability of the random-surfer walk on the reversed dependency graph~\cite{page1999pagerank}.
As in TextRank/LexRank, a node becomes important if it is pointed to by other important nodes, yielding an eigenvector-centrality-like notion of salience that reflects the \emph{global} link structure~\cite{mihalcea2004textrank,erkan2004lexrank}.
In our setting, this means a prerequisite anchor receives higher importance when it is a structurally central dependency target of many downstream queries, especially when those downstream queries themselves occupy central positions in the workflow.

\etitle{A path-sum view (global dependency influence).}
Equation~\eqref{eq:rw_restart_uniform} admits a standard expansion that makes the propagation effect explicit.
Let $u$ denote the uniform distribution over $V$, i.e., $u(v)=1/|V|$.
In vector form, \eqref{eq:rw_restart_uniform} can be written as
\[
r \;=\; (1-\beta)\,u \;+\; \beta\,P^\top r,
\]
which yields
\begin{equation}
\label{eq:path_sum_uniform}
r \;=\; (1-\beta)\sum_{\ell=0}^{\infty}\beta^\ell (P^\top)^\ell u.
\end{equation}
This expansion shows that $r$ reflects the \emph{global} dependency structure: importance is repeatedly propagated along reversed prerequisite links with geometric attenuation.
As a result, anchors that are structurally central in the dependency workflow obtain higher scores.

\etitle{Why a structural term in $\kw{TSI}$ is sensible under sparse recurrence.}
The stationary score $r(\cdot)$ induces a lightweight structural-centrality ordering from the prerequisite DAG, reflecting eviction externalities that are not captured by entry-local recency/frequency.
This view is \emph{consistent with} our one-hop proxy $\kw{dep}(\cdot)$ in the main text:
$\kw{dep}(q)$ conservatively aggregates the direct dependent-side request mass attributable to an anchor, and thus provides a reasonable first-order structural signal in $\kw{TSI}(\cdot)$.
When finer structural discrimination is desired, the propagation-based score $r(\cdot)$ can be used as an \emph{optional refinement} to account for indirect support paths mediated by the dependency structure, while remaining efficiently computable by power iteration.

\begin{algorithm}[t]
\caption{Topic Retrieval via Representative Index (ANN Shortlist + Gated Routing)}
\label{alg:topic_retrieval}
\DontPrintSemicolon
\KwIn{incoming request $q_t$; threshold $\tau$; shortlist size $K$}
\KwOut{routed topic id $Z_t$ (or $\emptyset$ if none passes)}
\KwData{representatives $\{r(s)\}$ for active topics and a topic index supporting \textsf{IndexQuery}$(\phi(q_t),K)$}

$\mathcal{S}_t \leftarrow \textsf{IndexQuery}(\phi(q_t), K)$\; \tcp*[r]{retrieve up to $K$ nearest representatives}
$Z_t \leftarrow \emptyset$\;
\ForEach{$s\in \mathcal{S}_t$}{
  \If{$\mathrm{sim}(\phi(q_t), r(s))\ge \tau$}{
    \If{$Z_t=\emptyset$ \textbf{or} $\mathrm{sim}(\phi(q_t), r(s))>\mathrm{sim}(\phi(q_t), r(Z_t))$}{
      $Z_t \leftarrow s$\;
    }
  }
}
\Return{$Z_t$}\;
\end{algorithm}

\section{Representative Selection and Maintenance for Topic Routing}
\label{app:topic_mgmt}

\stitle{Topic abstraction and maintained state.}
We organize the cache into \emph{topics}, where each active topic $s$ groups cached queries routed to the same topical cluster.
Online, each topic maintains a compact state:
(i) a \emph{member set} $\mathcal{M}(s)$ of resident queries currently assigned to $s$;
(ii) a \emph{representative embedding} $r(s)$ serving as the routing key; and
(iii) per-query signals required by replacement (e.g., recency and $\mathrm{TSI}(\cdot)$ in Section~\ref{sec:component-B}).
To enable sub-linear routing across topics, we additionally maintain a topic-level vector index over $\{r(s)\}$ that supports top-$K$ nearest-neighbor queries and per-topic representative updates.
Topic creation/deletion is handled by the main workflow; here we focus on retrieval and on maintaining the above state consistently under online cache updates.

\begin{algorithm}
\caption{Representative and Index Maintenance (TSI-max Anchor with Lazy Refresh)}
\label{alg:topic_index_maint}
\DontPrintSemicolon
\KwIn{topic id $s$; event \textsf{Insert}$(q)$ or \textsf{Evict}$(q)$; optional on-demand call \textsf{Refresh}$(s)$}
\KwOut{updated $r(s)$ and the topic-index entry of $s$}
\KwData{topic state $\{\mathcal{M}(s),\, r(s),\, \textsf{src}(s)\in\mathcal{M}(s)\ \text{s.t.}\ r(s)=\phi(\textsf{src}(s))\}$; per-query signal $\mathrm{TSI}(q)$ (Section~\ref{sec:component-B}); topic index with \textsf{IndexUpdate}$(s,r(s))$}

\SetKwProg{Fn}{Procedure}{}{}

\Fn{\textsf{Refresh}$(s)$}{
  \If{$\mathcal{M}(s)=\emptyset$}{
    delete topic $s$ from the index\;
    \Return\;
  }
  \If{$\textsf{src}(s)=\emptyset$ \textbf{or} $\textsf{src}(s)\notin \mathcal{M}(s)$}{
    $\textsf{src}(s)\leftarrow \arg\max_{e\in\mathcal{M}(s)} \mathrm{TSI}(q)$\;
    $r(s)\leftarrow \phi(\textsf{src}(s))$\;
    \textsf{IndexUpdate}$(s, r(s))$\;
  }
}

\Fn{\textsf{OnInsert}$(s,q)$}{
  append $q$ to $\mathcal{M}(s)$\;
  compute/update $\mathrm{TSI}(q)$\;
  \If{$\textsf{src}(s)=\emptyset$ \textbf{or} $\mathrm{TSI}(q)>\mathrm{TSI}(\textsf{src}(s))$}{
    $\textsf{src}(s)\leftarrow q$\;
    $r(s)\leftarrow \phi(q)$\;
    \textsf{IndexUpdate}$(s, r(s))$\;
  }
}

\Fn{\textsf{OnEvict}$(s,q)$}{
  remove $q$ from $\mathcal{M}(s)$\;
  \If{$\mathcal{M}(s)=\emptyset$}{
    delete topic $s$ from the index\;
    \Return\;
  }
  \If{$\textsf{src}(s)=q$}{
    $\textsf{src}(s)\leftarrow \emptyset$\tcp*[r]{invalidate; refreshed lazily by \textsf{Refresh}}
  }
}
\end{algorithm}

\stitle{Anchor-based representative (invariant).}
We adopt an \emph{anchor-based} representative: at any time, $r(s)$ equals the embedding of a \emph{single} resident query in $\mathcal{M}(s)$.
Among current members, the anchor is chosen as a query with maximal $\mathrm{TSI}(\cdot)$ (Section~\ref{sec:component-B}), with deterministic tie-breaking (e.g., by recency or a fixed id order).
This invariant avoids centroid-like summaries and yields a stable routing key under churn, while the TSI-max rule biases the anchor toward the topic core and structurally necessary context.

\stitle{Topic-level retrieval and routing (coarse stage).}
Given an incoming request $q_t$, we first use the representative index to retrieve a small candidate set of topics:
\[
\mathcal{S}_t \leftarrow \textsf{IndexQuery}(\phi(q_t), K),
\]
where \textsf{IndexQuery} returns up to $K$ topics whose representatives are nearest to $\phi(q_t)$ in the embedding space.
We then apply a similarity gate and select the best passing candidate:
\[
Z_t \leftarrow \arg\max_{s\in \mathcal{S}_t:\ \mathrm{sim}(\phi(q_t), r(s))\ge \tau}\ \mathrm{sim}(\phi(q_t), r(s)),
\]
and set $Z_t=\emptyset$ if no candidate passes the gate.
Algorithm~\ref{alg:topic_retrieval} summarizes this retrieval-and-gate routine.
(New-topic creation for $Z_t=\emptyset$ is handled in the main workflow and is not repeated here.)
This stage implements coarse-to-fine retrieval: the index constrains routing to a shortlist, while the gate preserves topic purity.

\stitle{Within-topic verification (fine stage).}
After routing, we decide semantic reuse by a local search restricted to the routed topic.
Concretely, we find the most similar resident query within $\mathcal{M}(Z_t)$:
\[
q_{\mathrm{hit}} \leftarrow \arg\max_{q\in \mathcal{M}(Z_t)} \mathrm{sim}\!\big(\phi(q_t), \phi(q)\big).
\]
A semantic hit is declared if $\mathrm{sim}\!\big(\phi(q_t), \phi(q_{\mathrm{hit}})\big)\ge \tau$ (or a stricter reuse threshold if routing and reuse gates are decoupled in the implementation).
This separates \emph{coarse} candidate-topic retrieval from \emph{fine} in-topic verification, while keeping verification localized to the routed topic.

\stitle{Representative and index maintenance.}
We maintain $r(s)$ and the topic-level index online under insertions and evictions, while preserving the anchor-based invariant.
Representative updates are event-driven: $r(s)$ changes only when a newly inserted query becomes the TSI-max anchor, or when the current anchor is invalidated by eviction.
When $r(s)$ changes, we update the corresponding entry in the topic-level index to keep retrieval consistent.
Algorithm~\ref{alg:topic_index_maint} summarizes the maintenance logic.

\etitle{On insertion.}
When a new query $q$ is assigned to topic $s$, we append it to $\mathcal{M}(s)$ and compute (or refresh) $\mathrm{TSI}(q)$.
If $\mathrm{TSI}(q)$ exceeds that of the query currently realizing $r(s)$, we update $r(s)\leftarrow \phi(q)$ and update the index entry of $s$; otherwise, $r(s)$ remains unchanged.
This yields $O(1)$ representative updates on insertions.

\etitle{On eviction (lazy refresh).}
When a query is evicted from topic $s$, we remove it from $\mathcal{M}(s)$.
If the evicted query is the one realizing $r(s)$, we refresh the representative lazily:
the refresh is deferred until the next time $s$ is needed as a routing candidate (e.g., returned by \textsf{IndexQuery}) or otherwise accessed, at which time we scan $\mathcal{M}(s)$ to select a query with maximal $\mathrm{TSI}(\cdot)$ and reset $r(s)$ to its embedding, followed by an index update.
This scan is amortized since it is triggered only when the current anchor query is evicted.

\etitle{Empty-topic handling.}
If $\mathcal{M}(s)=\emptyset$, we delete topic $s$ and remove it from the topic-level index, ensuring that the index only contains representatives of active topics.

\stitle{Retrieval-oriented discussion.}
The shortlist size $K$ trades off routing recall and per-request cost, while the threshold $\tau$ trades off topic purity and fragmentation.
Compared with centroid-like summaries, an anchor-based representative is cheaper to maintain and more stable under churn; selecting the anchor by TSI further biases $r(s)$ toward core and structurally necessary context, mitigating representativeness loss under mild intra-topic variability.

\end{document}